\newif\ifsubmit %
\newif\ifphone %

\submittrue

\ifsubmit %
\documentclass[11pt]{article} %
\else %
\ifphone %
\documentclass[11pt]{article} %
\else \documentclass{article} %
\fi \fi

\usepackage{article} %
\usepackage{math} %
\usepackage{algo} %
\usepackage{wrapfig} %
\usepackage[section]{placeins} %
\usepackage{graphicx} %
\usepackage{layout}   %
\usepackage[longnamesfirst,numbers]{natbib} %
\usepackage{setspace} %
\usepackage[center,big]{titlesec} %
\usepackage{hyperref}             %

\usepackage[font={small},labelfont=bf,textfont=it]{caption}

\ifsubmit %
\usepackage[letterpaper]{geometry} 
\else %
\ifphone %
\usepackage[paperwidth=6in,paperheight=12.99in,margin=1cm,includehead,includefoot]{geometry}
\fi \fi

\allowdisplaybreaks

\everymath{\displaystyle}

\newcommand{\defterm}[1]{{\boldmath\normalfont \bfseries #1}}

\newcommand{\uni}{\overline} %

\newcommand{\spn}{\operatorname{span}\optpar}
\newcommand{\convex}{\operatorname{convex}\optpar} %

\newcommand{\groundsets}{\subsetsof{\groundset}} %
\newcommand{\groundcube}{\bracketsof{0,1}^{\groundset}} %
\newcommand{\matroid}{\mathcal{M}} %
\newcommand{\groundset}{\mathcal{N}} %
\newcommand{\independents}{\mathcal{I}} %
\newcommand{\defmatroid}{%
  \matroid = \parof{\groundset,\independents}%
} %

\newcommand{\defnnsubmodular}{f: \groundsets \to \nnreals} %
\newcommand{\defrsubmodular}{f: \groundsets \to \reals} %

\NewDocumentCommand{\mlf}{O{F} g g}{%
  \IfNoValueTF{#3}{ #1\IfNoValueF{#2}{%
      \parof{#2}%
    }%
  }{%
    #1\parof{#3 + #2} - #1 \parof{#3}%
  }%
}%
\NewDocumentCommand{\dmlf}{o}{F'\IfNoValueF{#1}{_{#1}}\optpar }
\NewDocumentCommand{\ddmlf}{o o}{F'' %
  \IfNoValueF{#1}{ %
    \IfNoValueTF{#2}{ %
      _{#1} %
    }{ %
      _{#1,#2} %
    } }%
  \optpar } %

\NewDocumentCommand{\f}{O{f} g g}{
  #1\IfNoValueF{#3}{_{#3}}\IfNoValueF{#2}{\parof{#2}} }
\NewDocumentCommand{\ef}{O{f} g g}{
  \evof{#1\IfNoValueF{#3}{_{#3}}\IfNoValueF{#2}{\parof{#2}}} }
\providecommand{\eg}{\ef[g]} %
\NewDocumentCommand{\uS}{G{\infty}}{ \uni{S}_{#1} }
\NewDocumentCommand{\uR}{G{\infty}}{ \uni{R}_{#1} }
\newcommand{\uI}{\uni{I}\optsub} %
\newcommand{\uJ}{\uni{J}\optsub} %

\newcommand{\g}{\f[g]} %
\newcommand{\apxg}{\f[\tilde{g}]} %

\begin{document}

\title{Parallelizing~greedy for submodular~set~function~maximization
  in~matroids~and~beyond\footnote{This work is partially
    supported by NSF grant CCF-1526799. University of Illinois,
    Urbana-Champaign, IL 61801. {\tt
      \{chekuri,quanrud2\}@illinois.edu}.}
}

\author{Chandra Chekuri \and Kent Quanrud}

\maketitle

\begin{abstract}
  We consider parallel, or low adaptivity, algorithms for submodular
  function maximization. This line of work was recently initiated by
  Balkanski and Singer and has already led to several interesting
  results on the cardinality constraint and explicit packing
  constraints. An important open problem is the classical setting of
  matroid constraint, which has been instrumental for developments in
  submodular function maximization. In this paper we develop a general
  strategy to parallelize the well-studied greedy algorithm and use it
  to obtain a randomized $\parof{\frac{1}{2} - \eps}$-approximation in
  $\bigO{\frac{\log^2 n}{\eps^2}}$ rounds of adaptivity. We rely on
  this algorithm, and an elegant amplification approach due to
  \citet{bv} to obtain a fractional solution that yields a
  near-optimal randomized $\parof{1-1/e-\eps}$-approximation in
  \begin{math}
    \bigO{\frac{\log^2 n}{\eps^3}}
  \end{math}
  rounds of adaptivity.  For non-negative functions we obtain a
  $\parof{3-2\sqrt{2}}$-approximation and a fractional solution that
  yields a $\parof{\frac{1}{e} - \eps}$-approximation.  Our approach
  for parallelizing greedy yields approximations for intersections of
  matroids and matchoids, and the approximation ratios are comparable
  to those known for sequential greedy.
\end{abstract}

\newpage

\section{Introduction}

Matroids and submodular functions are two fundamental objects in
combinatorial optimization that help generalize and unify many
results. A matroid $\defmatroid$ consists of a finite ground set
$\groundset$ and a downclosed family of independent sets
$\independents \subseteq \groundsets$ that satisfy a simple exchange
property. Whitney \cite{w-35} defined matroids to abstract properties
of dependence from linear algebra.  Matroids are surprisingly common
in combinatorial optimization and capture a wide variety of
constraints. Submodular set functions are a class of real-valued set
functions $\defrsubmodular$ that discretely model decreasing marginal
returns. (Formal definitions for both matroids and submodular
functions are given in \refappendix{preliminaries}.)  More than
purely mathematical generalizations, matroids and submodular functions
capture the computational character of their more concrete
instances. One can maximize a linear function over a matroid, the same
way that one can compute a maximum or minimum weight spanning
tree. For a submodular function $f$, one can obtain an optimal
$\parof{1 - e^{-1}}$-approximation to maximizing $f(S)$ over a
cardinality constraint on $S$ \cite{nwf-78}, the same way that one can
obtain an optimal $\parof{1 - e^{-1}}$-approximation to maximum
coverage subject to a cardinality constraint (where optimality assumes
$\text{P} \neq \text{NP}$) \cite{f-98}. Both of these connections are
realized by a simple greedy algorithm.

The above connections are by now classical.  A comparably recent
result is an optimal $\parof{1 - e^{-1}}$-approximation to maximizing
$f(S)$ over a matroid constraint $S \in \independents$ when $f$ is
monotone \citep{ccpv}, a significant generalization of the cardinality
constraint problem.  Subsequent developments obtained a
$e^{-1}$-approximation for nonnegative submodular functions subject to
a matroid constraint \citep{fns} and then improved beyond $e^{-1}$
\citep{en-16,bf-16}. The techniques underlying these results take a
fractional point of view with one part continuous optimization and
another part rounding, somewhat analogous to the use of LP's for
approximating integer programs. These and other techniques lead to a
number of improved approximations for other set systems, such as
combinations of matroids and packing constraints. There has been
significant work on submodular function maximization in the recent
past based also on classical combinatorial techniques such as greedy
and local search and some novel variants such as the double greedy
algorithm which led to an optimum $\frac12$-approximation for
unconstrained maximization of a nonnegative submodular function
\cite{bfns-15}.

With generality comes applicability. Many modern problems in machine
learning and data mining can be cast as submodular function
maximization subject to combinatorial or linear constraints. These
applications leverage large amounts of informative data that form very
large inputs in the corresponding optimization problem.  Practical
concerns of scalability have lead to theoretical consideration of
various ``big data'' models of computation, such as faster
approximations \citep{bv,bfs,mbkvk,mbk-16,fhk,en-17}, online and
streaming models
\citep{bmkk,ck-15,bfs-15,cgq,mbk-16,chjkt,hky,mjk-18,fkk-18,ntmzms},
and algorithms conforming to map-reduce frameworks
\citep{mksk-16,kmvv-15,benw-15,mz-15,benw-16,lv}.

To retain the full generality of submodular functions, submodular
optimization algorithms are generally specified in a value oracle
model, where one is given access to an oracle that returns the value
$f(S)$ for any set $S$. Algorithms in this model are typically
measured by the number of oracle calls made in addition to usual
computational considerations. From a learning theoretic perspective,
the oracle model raises a question of how much one can learn about a
submodular function from a limited number of queries to the oracle.
\citet{brs-17} showed that a polynomial number of queries of random
samples (drawn from any distribution) cannot lead to a constant factor
approximation for maximizing a submodular function subject to a
cardinality constraint. The key point is that the queries are selected
independently. All the standard polynomial time algorithms, of course,
make only polynomially many queries, but these queries are chosen
sequentially, one query informing the next.

Motivated by both practical and theoretical interests, \citet{bs-18}
investigated the minimum required \defterm{adaptivity} for submodular
maximization. Adaptivity can be interpreted (to some extent) as a
parallel model of computation with parallel access to oracles. For
$k \in \naturalnumbers$, an algorithm with access to an oracle is
\defterm{$k$-adaptive} if all of its oracle queries can be divided
into a sequence of $k$ rounds, where the choice of queries in one
round can only depend on the values of queries in strictly preceding
rounds. For example, the greedy algorithm for cardinality constraints
has adaptivity equal to the specified cardinality. The negative
results of \citep{brs-17} can be restated as saying more than 1 round
of adaptivity is necessary for a constant approximation. Pushing this
model further, \citet{bs-18} showed that no algorithm can obtain
better than a $\bigO{1 / \log n}$ approximation ratio with less than
$\bigO{\log{n} / \log \log n}$ adaptivity. On the positive side,
\citet{bs-18} gave a $\frac{1}{3}$-approximation for maximizing a
monotone submodular function using $\bigO{\log n}$ adaptive rounds.
This was improved by $\parof{1 - e^{-1} - \eps}$-approximation
algorithm for maximizing a monotone submodular function subject to a
cardinality with $\bigO{\log{n} / \poly{\eps}}$ adaptivity
\citep{brs-19,en-19}. A number of follow up works have extended these
results to knapsack and packing constraints \citep{cq-19,env-18-1} and
nonnegative submodular functions \citep{bbs-18,fmz-18-nn,env-18-1}, or
improved other aspects such as the total number of oracle calls
\citep{fmz-18-monotone}. One could argue that all of these papers
essentially build upon the understanding and analysis for the
cardinality constraint (even the ones which address significantly more
general packing constraints \citep{cq-19,env-18-1}).

One classical setting that is important to address is submodular set
function maximization subject to an arbirary matroid constraint. Let
$\defmatroid$ be a matroid, and $\defnnsubmodular$ a nonnegative
submodular function. The goal is to compute, in parallel, an
independent set $I \in \independents$ maximizing $f(I)$. We note that
the natural and simple greedy algorithm gives a $1/2$-approximation
when $f$ is monotone \cite{fnw-78}.  A strong theoretical motivation
to study this problem is to understand the extent to which the
classical greedy algorithm that gives good approximation for matroid
constraints, and several generalizations, can be
parallelized. Historically, the matroid constraint problem has been
important for developing several new algorithmic ideas including the
multilinear relaxation approach \cite{ccpv}.  Before giving our
results, it is important to establish the model, and in particular how
we engage the matroid from a parallel perspective.

As with submodular functions, algorithms optimizing over matroids
typically access the matroid structure via oracles. Standard oracles
for a matroid $\defmatroid$ are \defterm{independence oracles}, which
take as input a set $S \subseteq \groundset$ and return whether or not
$S \in \independents$; \defterm{rank oracles}, which take as input a
set $S \subseteq \groundset$ and return the maximum cardinality of any
set in $S$; and \defterm{span oracles}, which take as input a set
$S \subseteq \groundset$ and an element $e \in \groundset$ and returns
whether or not $S + e$ has higher rank than $S$.  Rank oracles are
stronger than both independence oracles (since a set $S$ is
independent iff $\rank{S} = \sizeof{S}$) and span oracles (by
comparing $\rank{S}$ and $\rank{S+e}$).  In this work, we assume
access to span oracles, and extend the notion of adaptivity to span
oracles in the natural way.

Parallel rank oracles are known for most useful matroids. Parallel
rank oracles for graphic matroids are given by parallel algorithms for
computing spanning trees, such as Borůvka's algorithm. Rank oracles
for linearly representable matroids (i.e., independent sets of vectors
in some field) are also known
\citep{imr,bgh,chistov,mulmuley}. 
We note that many standard matroids (such as partition matroids and
graphic matroids) can be viewed as linearly representable matroids.
Parallel oracle models are well established in the literature.  For
example, the parallel oracle model was assumed by \citet{kuw-88}, who
considered parallel algorithms w/r/t both independence and rank
oracles for computing maximal and maximum independent sets,
generalizing work on perfect matchings \citep{kuw-86}. The oracle
model was also considered by \citet{nsv}, who obtained parallel
algorithms for matroid union and intersection in representable
matroids, and asked if similar results can be obtained for general
matroids assuming access to rank or independence oracles.

We are now prepared to state our results. We first obtain results for
matroid constraints that are competitive with the best known
sequential algorithms and are polylogarithmically adaptive. We obtain
different approximation ratios depending on whether one desires a
discrete independent set or a fractional point in the independent set
polytope. The fractional point is evaluated w/r/t the ``multilinear
extension'' $F: \nnreals^{\groundset} \to \reals$ of the set function
$f$. The multilinear extension (defined in \refappendix{submodular})
is a continuous extension of $f$ first applied to submodular
optimization in \citep{ccpv}.  A point in the independent set polytope
of a matroid can be rounded to a discrete set without loss, and nearly
all competitive approximation algorithms for matroid constraints are
obtained by approximating the multilinear extension and then rounding
\citep{ccpv,fns,cvz-10,bv} (with the notable exception being
\citep{fw}).  The rounding schemes, however, are not known to be
parallelizable for general matroids. For now, the fractional solutions
give a certificate that allow us to approximate the optimum value.  In
the following, let $\opt = \max_{I \in \independents} f(I)$ denote the
maximum value of any independent set. We first give results for
maximizing a monotone submodular function subject to a matroid
constraint.

\begin{theorem}
  \labeltheorem{monotone-matroid}
  Let $\defmatroid$ be a matroid and $\defnnsubmodular$ a monotone
  submodular function.
  \begin{itemize}
  \item   There is a randomized algorithm that
  outputs a set $I \in \independents$ s.t.\
    \begin{math}
      \ef{I} \geq \parof{\frac{1}{2} - \eps} \opt
    \end{math}
   and has adaptivity $K$ where $\evof{K} = \bigO{\frac{\log^2 n}{\eps^2}}$.
 \item
 There is a randomized algorithm that computes a convex combination of
  $\bigO{\reps}$ independent sets $x \in \convex{\independents}$ s.t.\
    \begin{math}
      \mlf{x} \geq \parof{1 - e^{-1} - \eps} \opt
    \end{math}
    in $\bigO{\frac{\log^2 n}{\eps^3}}$ expected adaptive rounds,
    which implies that one can compute a $\parof{1 - e^{-1} - \eps}$ approximation
    to $\opt$.
  \end{itemize}
\end{theorem}

\begin{remark}
  We state the approximation and adaptivity bounds as expected
  quantities. We can achieve a high-probability
  bounds via standard tricks, however, in the interest of clarity we
  leave these details for a future version.
\end{remark}

We also obtain approximations for generally nonnegative submodular
functions with low adaptivity.

\begin{theorem}
  \labeltheorem{nn-matroid}
  Let $\defmatroid$ be a matroid and $\defnnsubmodular$ a nonnegative
  submodular function.
  \begin{itemize}
  \item There is a randomized algorithm that
computes an independent set $I \in \independents$ such that
  $\ef{I} \geq \epsless \parof{3 - 2 \sqrt{2}} \opt$ and
  has adaptivty $\bigO{\frac{\log^2 n}{\eps^2}}$ in expectation.
\item There is a randomized algorithm that computes a convex combination of
  $\bigO{\reps}$ independent sets
  $I_1,\dots,I_{\ell} \in \independents$
   with $\bigO{\frac{\log^2 n}{\eps^3}}$ adaptive rounds
  in expectation such that, if $J_{k}$ samples each element in $I_k$
  independently with probability $1/\ell$ for each $k \in [\ell]$, we
  have $ \ef{J_1 \cup \cdots \cup J_{\ell}} \geq \parof{e^{-1} - \eps} \opt$.
  This implies that one can compute a $\parof{e^{-1} - \eps}$ approximation
  to $\opt$.
  \end{itemize}

\end{theorem}

\begin{remark}
  Given a fractional point $x$ in the matroid polytope it can be
  rounded to an independent set $I$ such that $\evof{f(I)} = F(x)$
  when $f$ is submodular. Pipage rounding \cite{ccpv} and swap
  rounding \cite{cvz-10} achieve this. Swap rounding requires the
  decomposition of $x$ into a convex combination of independent sets
  and consists of repeatedly (randomly) merging two independent sets
  via exchanges.  The algorithms in the preceding theorems provide
  such a decomposition with only $O(1/\eps)$ independent sets. For
  some simple matroids such as partition matroids one can implement
  the random merging between two independent sets in parallel rather
  easily.  This suggests an interesting open problem: for which
  matroids can two independent set be merged randomly in parallel?
\end{remark}

The techniques extend beyond matroids to combinations of matroids,
such as matroid intersections or (more generally) $p$-matchoids. These
systems generalize bipartite matchings, arboresences, and
non-bipartite matchings, and are defined formally in
\refappendix{matroid-combinations}. W/r/t the oracle model, we note
that these set systems are defined by some underlying collection of
matroids, and we assume access to a rank oracle for each underlying
matroid. We also note that, unlike in matroids, rounding a fractional
solution in matroid intersections and matchoids incurs additional
constant factor loss in the approximation, but we still state bounds
for the multilinear relaxation as they are of independent interest and
can be used in contention resolution schemes when combining with other
constraints \cite{cvz-14}.

\begin{theorem}
  \labeltheorem{monotone-matchoids}
  Let $\defmatroid$ be a $p$-matchoid for some
  $p \in \naturalnumbers$, and let $\defnnsubmodular$ be a monotone
  submodular function. There is a randomized algorithm that computes
  an independent set $I \in \independents$ s.t.\
  $\evof{f(I)} \geq \prac{1 - \eps}{p+1} \opt$ with
  $\bigO{\frac{\log^2 n}{\eps^2}}$ adaptive rounds in expectation.
  There is a randomized algorithm that computes a
  convex combination of $\bigO{\reps}$ independent sets
  $x \in \convex{\independents}$ s.t.\ $\evof{\mlf{x}} \geq \parof{1 - e^{-1/p} - \eps} \opt$
     with $\bigO{\frac{\log^2 n}{\eps^3}}$ adaptive rounds
    in expectation.
\end{theorem}

For nonnegative functions we have the following theorem.
\begin{theorem}
  \labeltheorem{nn-matchoids}
  Let $\defmatroid$ be a $p$-matchoid for some
  $p \in \naturalnumbers$, and let $\defnnsubmodular$ be a nonnegative
  submodular function. There is a randomized algorithm that computes
  $I \in \independents$ such that
  $\ef{I} \geq \epsless \prac{1 +o(1)}{4(p+1)} \opt$ with
  $\bigO{\frac{\log^2 n}{\eps^2}}$ adaptive rounds in expectation. One
  can computed $\ell = \bigO{1/\eps}$ randomized independent sets
  $I_1,\dots,I_{\ell} \in \independents$
   with $\bigO{\frac{\log^2 n}{\eps^3}}$ adaptive rounds
  in expectation such that, if $J_{k}$ samples each element in $I_k$
  independently with probability $1/\ell$ for each $k \in [\ell]$, we
  have $ \ef{J_1 \cup \cdots \cup J_{\ell}} \geq \frac{1 - \eps}{p}
    e^{-1/p} \opt$.
  \end{theorem}

\subsection{Overview of techniques}

\labelsection{overview}

\FloatBarrier

We give a brief overview of the techniques that lead to our
results. The overall algorithm is relatively simple, and much shorter
to describe than to analyze fully. Moreover, it is a composition of
modular techniques, each of which may be of independent interest.  For
the sake of discussion, we focus on the setting of maximizing a
monotone submodular function subject to a matroid constraint, noting
that the techniques apply to nonnegative submodular functions as well.

Consider the simple greedy algorithm in the sequential setting, given
in \reffigure{greedy}. The greedy algorithm starts with an empty set
$S = \emptyset$, and greedily adds to $S$ the element $e$ maximizing
$f(S + e)$ subject to $S + e \in \independents$. It terminates when
there are no longer any elements that can be added to $S$.
\begin{figure}
  \centering
  \begin{algorithm}[.5\paperwidth]{greedy}{$\defmatroid$,$\defnnsubmodular$}
  \item $S \gets \emptyset$
  \item while $S$ is not a base
    \begin{steps}
    \item $e \gets \argmax[d]{f_S(d) \where I + d \in \independents}$
    \item $S \gets S + e$
    \end{steps}
  \item return $S$
  \end{algorithm}
  \begin{implicitframed}[.5\paperwidth]
    \caption{A greedy $2$-approximation for maximizing a monotone
      submodular function subject to a matroid
      constraint.\labelfigure{greedy}}
  \end{implicitframed}
\end{figure}
The greedy algorithm has an approximation ratio of $1/2$, short of the
optimal $(1-e^{-1})$-approximation scheme. This suboptimal ratio is
not the primary concern because \textit{a priori} it is not clear how
to get any constant factor approximation for matroids in
parallel. (Moreover, a $1/2$ approximation ratio can be amplified to
$1-e^{-1}$; more on this later.)  The greater problem in our setting
is that it is inherently very sequential. The number of iterations,
$\rank{\matroid}$, can be as large as $n$. One hopes to ``flatten''
the iterations, but each chosen element $e$ depends on the previously
selected elements $S$ at two critical points: (a) the marginal value
$f_S(e)$ decreases as $S$ increases due to submodularity, and (b) we
cannot take $e$ if $S + e$ is infeasible. Note that when one considers
a cardinality constraint the second issue is significantly less
important since every element in the ground set can be feasibly added
as long as we have not reached a base. In fact all recent papers on
adaptivity are focused mainly on the first issue.

It is clear that to obtain $\poly{\log n, 1/\eps}$ depth, we cannot
take just one element at a time, and would prefer to somehow take,
say, $\rank{\matroid} / \poly{\log n, 1/\eps}$ elements in each
parallel round. In fact, we do not even take a set per iteration, but
pairs of sets that we call ``greedy blocks''.  For any set system
$\defmatroid$, a \defterm{$\epsless$-greedy block} consists of a
random pair of sets $(I,S)$ s.t.\
\begin{enumerate}
\item $I \subseteq S$ and $I \in \independents$.
\item
  \begin{math}
    \evof{f(I)} \geq \epsless \evof{\sizeof{S}} \max{0,\max_{e} f(e)}.
  \end{math}
\end{enumerate}
We are interested in randomized greedy blocks not only w/r/t
$\matroid$ and $f$, but the contracted matroid $\matroid / Q$ and
marginal values $f_Q$ for various sets $Q \subseteq \groundset$.

Note that $S$ is not required to be independent, but we still require
an independent certificate $I \subseteq S$ that captures most of the
value of $S$.  At the end of the day, we will output a union of the
independent $I$-components of the greedy blocks which will form a
feasible solution. Allowing $S$ to be dependent is an important degree
of freedom that is used to make measurable progress and bound the
depth.

We produce greedy blocks for matroids by a simple ``greedy sampling''
procedure. Let $\lambda \geq \max_{e \in \groundset} f(e)$ be some
upper bound on the margin of any element.  For a carefully chosen
value $\delta > 0$ we let $S$ sample each element $e$ with
nearly-maximum marginal value $f(e) \geq \epsless \lambda$
independently with probability $\delta$. We then prune any sampled
element that is either (a) spanned by the other sampled elements or
(b) has marginal value $\leq \epsless \lambda$ w/r/t the other sampled
elements. The pruning step leaves an independent subset
$I \subseteq S$ where each element contributes at least
$\epsless \lambda$ to $I$. We choose $\delta$ conservatively so that
$I$ retains most of the value of $S$, but also as large as possible
within these constraints. This ``greedy'' choice of $\delta$ ensures
that the residual problem (consisting of large margin elements not
spanned by $S$) is smaller by about an $\eps$-fraction in
expectation. We can both search for the appropriate value of $\delta$
and sample with probability $\delta$ in parallel.  The basic idea of
greedy sampling is directly inspired by a much simpler greedy sampling
procedure in the cardinality setting in our previous work
\citep{cq-19}.\footnote{There one chooses $\delta$ as large as
  possible subject to preserving the gradient (along high margin
  elements) of the multilinear extension on average. The initial
  inspiration for the greedy sampling procedure for matroids was to
  apply the same logic to the rank function of the matroid, which is
  also a monotone submodular function.}

We now iterate along greedy blocks, where each iteration is w/r/t the
residual system induced by previously selected greedy blocks. We start
with empty sets $I,S = \emptyset$. We repeatedly generate
$\epsless$-greedy blocks $(I',S')$ w/r/t the contracted matroid
$\matroid / S$ and the marginal values $f_S$, and then add $I'$ to $I$
and $S'$ to $S$. The $\lambda$'s are decreased by multiplicative
factors of $1-\eps$ to bound the depth: within a fixed $\lambda$, we
expect a limited number of greedy samples until there are no elements
left of marginal value $\geq \epsless \lambda$; $\lambda$ can decrease
by $1-\eps$ multiplicative factor only a limited number of times
before the marginal values of remaining elements are negligibly
small. When $\lambda$ is small enough, we return $I$, which is an
independent set.

The above produces a randomized independent set $I$ w/
$2\ef{I} \geq \epsless \opt$. In fact, it achieves a slightly stronger
bound of the form $\ef{I} \geq \epsless f_S(T)$ for any independent
set $T \in \independents$, where $S$ is a set containing $I$ such that
$\ef{S} \leq \epsmore \ef{I}$. This is similar to (but slightly weaker
than) the sequential greedy algorithm, which outputs an independent
set $I$ such that $f(I) \geq f_I(T)$ for any independent set
$T \in \independents$.

To improve the bound of $1/2$ we rely on an elegant result of
\citet{bv}. Motivated by the problem of finding a faster approximation
algorithms for submodular function maximization, they showed that
$\bigO{1/\eps}$ iterations of the greedy algorithm via auxiliary
functions induced by the multilinear relaxation produces a convex
combination of $\bigO{1/\eps}$ independent sets with approximation
ratio $\parof{1 - e^{-1} - \eps}$. We call this process ``multilinear
amplification'', and show that it can be used, via our $1/2$
approximation, to produce $\parof{1 - e^{-1} - \eps}$-approximate
fractional solution.

\subsection{Further discussion and related work}
Submodular function maximization is a classical topic that was
explored in several papers in the 70's with influential papers on the
performance of the greedy algorithm, in particular by the work of
\citet{nwf-78} and several subsequent ones. Recent years have seen a
surge of activity on this topic. There have many important and
interesting theoretical developments ranging from algorithms, lower
bounds, connections to learning and game theory, and also a number of
new applications in various domains including machine learning and
data mining.  Greedy and local-search methods have been revisited and
improved and approximation algorithms for nonnegative nonmonotone
functions were developed. An important development was the
introduction of the multilinear relaxation approach which brought
powerful continuous optimization methods into play and led to a new
algorithmic approach.  The literature is too large to discuss in this
paper. We are primarily motivated by the work in approximation
algorithms and the recent interest in finding parallel algorithms.  We
refer the reader to \cite{bs-18} for background and motivation
studying the notion of adaptivity which extends beyond interest in
parallelization.  Studying adaptivity allows one to avoid certain
low-level details of the traditional PRAM model. However, we note that
modulo poly-logarithmic factors and the reliance on an oracle for $f$
and the rank function of the the underlying matroids, all of our
algorithms can be implemented in the PRAM model without much effort.
Parallel algorithms for Set Cover have been well-studied earlier (see
\cite{bpt-11} and references) and more recent work has explored Max
Coverage and submodular function maximization in modern parallel
systems such as the Map-Reduce model \cite{kmvv-15,benw-15,lv}. The
study of parallelism for abstract problems that we consider here can
provide high-level tools that could be specialized and refined for
various concrete problems of interest.

Some of these results have been obtained
independently. \citet{brs-new} obtained similar results for monotonic
submodular functions, corresponding to \reftheorem{monotone-matroid}
and \reftheorem{monotone-matchoids}.  \citet{env-18-2} (updating an
earlier manuscript \citep{env-18-1}) obtains similar results for
optimizing the multilinear relaxation subject to a matroid constraint,
both monotone and nonnegative, similar to
\reftheorem{monotone-matroid} and \reftheorem{nn-matroid}.

\paragraph{Organization:}
We employ standard methodology and definitions that can be found in
standard references in combinatorial optimization such as
\cite{schrijver-book}. For the sake of completeness, preliminaries are
given (at a leisurely pace) in
\refappendix{preliminaries}. \refsection{greedy-sample} describes and
analyzes the greedy sampling scheme which is the critical piece in our
algorithms. \refsection{greedy-blocks-apx} describes how greedy
sampling can be used iteratively to derive approximation algorithms
with low adaptivity. In \refsection{multilinear-amplification} we
describe the multilinear amplification idea from \cite{bv}, extend it,
and use it on top of the greedy algorithms to derive improved
bounds. In \refsection{estimation}, we fill in some low-level details
regarding sampling and estimation that were abstracted out of previous
sections.

\paragraph{Notation:}
We use the following notation. For a set $S$ and an element $e$, we
let $S+e$ denote the union $S \cup \setof{e}$ and let $S - e$ denote
the difference $S \setminus \setof{e}$. For a sequence of sets
$S_1,\dots,S_k$, we denote their union by
$\uS{k} = \bigcup_{i=1}^k S_i$.

\section{Greedy sampling}

\labelsection{greedy-sample}

\newcommand{\Rank}{\operatorname{Rank}\optpar} %

\NewDocumentCommand{\dRank}{g g}{ %
  \Rank'                        %
  \IfNoValueF{#1}{              %
    \IfNoValueTF{#2}{           %
      \parof{#1}                %
    }{                          %
      _{#1}\parof{#2}           %
    }                           %
  }                             %
}                               %

\begin{figure}
  \begin{algorithm}{greedy-sample}{$\defmatroid$, $f$, $\lambda$,
      $\eps$}
    \commentitem{Assume: $\epsless \lambda \leq f(e) \leq \lambda$ for
      all $e$} %
    \commentitem{Goal: randomized pair of sets $(I,S)$, where $I$ is
      an independent subset of $S$, satisfying
      \reftheorem{greedy-sample}}
  \item choose $\delta > 0$ large as possible s.t.\ for
    $S \sim \delta \groundset$, \labelstep{gms-step-size}
    \begin{steps}
    \item \labelstep{gms-step-size-margin} %
      \begin{math}
        \evof{\sizeof{\setof{e \where f_S(e) \leq \epsless
              \lambda}}} %
        \leq \eps \sizeof{\groundset}
      \end{math}
    \item \labelstep{gms-step-size-span}
      \begin{math}
        \evof{\sizeof{\spn{S}}} \leq \eps \sizeof{\groundset}
      \end{math}
    \end{steps}
  \item sample $S \sim \delta \groundset$
  \item
    $I \gets \setof{e \in S \where f_{S - e}(e) \geq \epsless \lambda,
      e \notin \spn{S - e}}$
  \item
    return $(I,S)$
  \end{algorithm}
  \begin{implicitframed}
    \caption{A greedy sampling procedure for generating greedy blocks in
      independence systems.\labelfigure{greedy-sample}}
  \end{implicitframed}
\end{figure}

In this section, we define and analyze a randomized procedure for
generating greedy blocks, called \algo{greedy-sample} and given in
\reffigure{greedy-sample}. A brief sketch was given in
\refsection{overview}, which we supplement now with a more thorough
description.

We focus on the setting where all elements $e \in \groundset$ have
margin $\epsless \lambda \leq f(e) \leq \lambda$ for some
$\lambda >0$.  For a carefully chosen value $\delta > 0$, we sample
each element in the ground set independently with probability $\delta$
to produce a set $S$. $S$ may be dependent. We prune any element
$e \in S$ that is spanned by $S - e$, leaving an independent subset
$I$.

As $\delta$ increases, and $S$ grows with it, submodularity pushes
down the expected value of $S$ per sampled element, and the expected
span of $S$ increases. In turn, the likelihood of pruning increases,
and the ratio $f(I) / \sizeof{S}$ decreases in
expectation. Conversely, as we take $\delta$ down to 0, $I$ converges
to $S$ in expectation, and for sufficiently small $\delta$, $(I,S)$ is
a greedy block. We define the ``sufficiently small'' threshold to be
such that in expectation, (a) at least $\epsless$-fraction of
$\groundset$ has marginal value $\geq \epsless \lambda$ w/r/t $S$, and
(b) at least $\epsless$-fraction of $\groundset$ is not spanned by
$S$. Any $\delta$ satisfying both (a) and (b), it is shown, produces a
greedy block $(I,S)$.

We greedily maximize $\delta$ s.t.\ the above constraints for the sake
of efficiency. Maximality ensures we reach the breaking point of one
of the two limiting conditions. If condition (a) is tight, when many
of the elements have their marginal value drop below
$\epsless \lambda$. If $\delta$ is large enough that (b) is tight,
then a large fraction of $\groundset$ should be sampled by $S$.
sample spans about an $\eps$-fraction of the entire matroid. Either
way, we expect to substantially decrease the number of un-spanned
elements with marginal value $\geq \epsless \lambda$.

The first condition, \refstep{gms-step-size-margin}, says we do not
sample past the point where the margins are decreasing by a lot. A
simpler form appears in previous work in the monotone cardinality
setting \citep{cq-19}.  The second condition
\refstep{gms-step-size-span}, is a more significant departure from the
cardinality setting, and strikes a balance between trying to span many
elements, and not having to prune too many of the sampled elements. To
build some intuition for condition \refstep{gms-step-size-span},
consider the graphic matroid on a ``fat path graph'', the multigraph
consisting of $k$ copies of each edge.
\begin{center}
  \includegraphics[height=4em,width=20em,keepaspectratio]{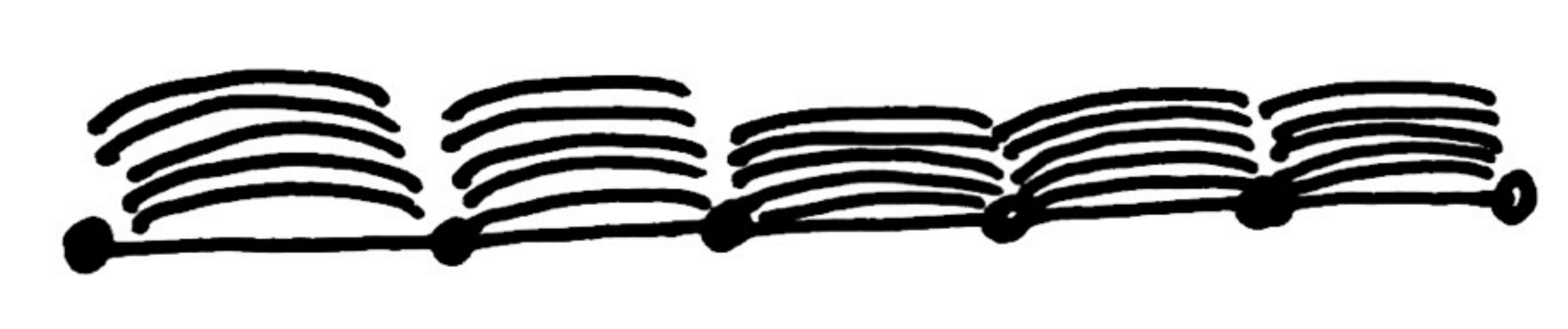}
\end{center}
Then \refstep{gms-step-size-span} (approximately) implies that
$\delta \leq \eps / k$. For this value of $\delta$, we expect a random
sample to span an $\eps$-fraction of the legs, hence an
$\eps$-fraction of the edges. We drop edges from legs that are sampled
multiple times, but the probability of double sampling from any
particular leg is $\leq \bigO{\eps^2}$.

For a second example where the distribution is less uniform, consider
the ``fat tail graph'', which is a multigraph where the first leg has $k$
copies of that edge for $k$ much larger than $1/\eps$.
\begin{center}
  \includegraphics[height=4em,width=20em,keepaspectratio]{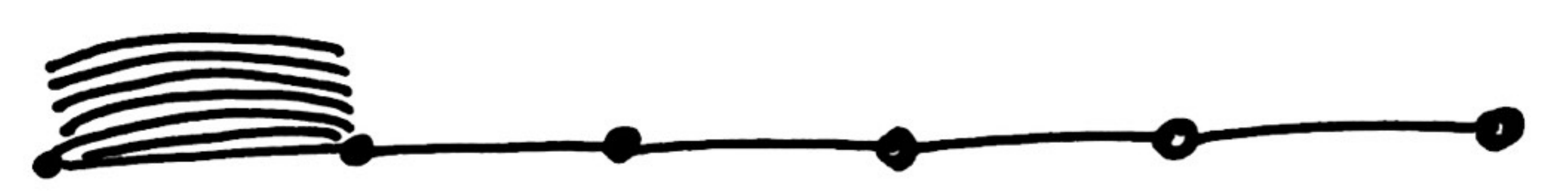}
\end{center}
We consider two regimes for $k$.  On one hand, if $k$ is comparable to
$n$, then \refstep{gms-step-size-span} is something like
$\delta \leq \frac{\eps}{k} \approx \frac{\eps}{n}$. A uniform sample
with probability $\delta$ will span the fat tail, hence a large
fraction of all the edges, with probability about $\eps$. The
probability of double sampling from the fat tail remains small.

More illuminating is the case where $k$ is much smaller than $n$. Then
\refstep{gms-step-size-span} roughly comes out to $\delta \leq
\eps$. We expect the uniform sample to span an $\eps$-fraction of the
single-edge legs, and almost certainly prune from the fat tail. That
is, the algorithm deliberately oversamples the fat tail. The edges
lost from oversampling fit into ``an $\eps$ of room'': the expected
number of edges pruned from the fat tail, $\eps k$, is much smaller
than the expected number of edges sampled, $\eps n$, and the loss is
essentially negligible.

\paragraph{Finding $\delta$ and other implementation issues:} All of
our algorithms critically rely on randoimization explicitly or
implicitly. For instance it is not obvious how to find the $\delta$ in
\refstep{gms-step-size} of \algo{greedy-sample}.  We rely on random
sampling and concentration bounds to estimate various quantities of
interest. A formal analysis of the estimation errors and how they
influence the approximation and adaptivity clutters the flow of the
main ideas. For this reason we relegate some of the implementation
details to \refsection{estimation}, and focus on the key high-level
steps.

\begin{theorem}
  \labeltheorem{greedy-sample} %
  Let $\defmatroid$ be a matroid or $p$-matchoid, let
  $\defrsubmodular$ be a submodular function, and let $\lambda \geq 0$
  such that $\epsless \lambda \leq f(e) \leq \lambda$ for all
  $e \in \groundset$. Then \algo{greedy-sample}($\matroid$, $f$,
  $\lambda$, $\eps$) returns random sets $(I,S)$ such that
  \begin{mathresults}
  \item $(I,S)$ is a $\apxless$-greedy block,
  \item
    \begin{math}
      \evof{\sizeof{\setof{e \notin \spn{S} \where f_S(e) \geq
            \epsless \lambda}}} \leq %
      \epsless \sizeof{\groundset}.
    \end{math}
  \end{mathresults}
\end{theorem}

\begin{proof}
  \begin{enumerate}[{label={(\roman*)}}]
  \item Clearly, $I$ is an independent subset of $S$. Indeed, if $I$
    is not independent, then there is $e \in I$ s.t.\
    $e \in \spn{I - e}$. But $\spn{I - e} \subseteq \spn{S - e}$ and
    $e \notin \spn{S - e}$ by choice of $e$.  We relate $\evof{f(I)}$
    and $\sizeof{S}$ via two intermediate sets.  Let
    \begin{center}
      \begin{math}
        Q = \setof{e \in S \where f_{S - e}(e) \geq \epsless \lambda}
      \end{math}
      and
      \begin{math}
        P = \setof{e \in S \where e \notin \spn{S - e}}.
      \end{math}
    \end{center}
    Then $I = P \cap Q$.  We claim that
    \begin{enumerate}[label={(\alph*)}]\itshape
    \item \begin{math}                                       %
        \evof{f(Q)} \geq \epsless \lambda \evof{\sizeof{S}}. %
      \end{math}                  %
    \item
      \begin{math}
        \evof{\sizeof{S \setminus P}} \leq \eps \evof{\sizeof{S}}.
      \end{math}                  %
    \end{enumerate}
    Assuming claims (a) and (b) hold, we have
    \begin{align*}
      \evof{f(I)}                 %
      &\tago{=}                           %
        \evof{f(Q)} - \evof{f_I(Q)} %
        \tago{\geq}                        %
        \evof{f(Q)} - \sum_{e \in Q} \evof{f_I(e)}
      \\
      &\tago{\geq}                        %
        \evof{f(Q)} - \sum_{e \in Q} \probof{e \in Q \setminus I} f(e)
        \tago{\geq}                        %
        \evof{f(Q)} - \lambda \sum_{e \in Q} \probof{e \in Q \setminus I}
      \\
      &=                        %
        \evof{f(Q)} - \lambda \evof{\sizeof{Q \setminus I}} %
        \tago{\geq}                                         %
        \evof{f(Q)} - \lambda \evof{\sizeof{S \setminus P}} %
      \\
      &\tago{\geq}
        (1-2\eps) \lambda \evof{\sizeof{S}}
    \end{align*}
    by \tagr $I \subseteq Q$, (\tagr*,\tagr*) submodularity, \tagr
    choice of $\lambda$, \tagr
    $Q \setminus I \subseteq S \setminus P$, and \tagr plugging in (a)
    and (b). It remains to prove claims (a) and (b).
    \begin{enumerate}[label={(\alph*)}]
    \item For each element $e$, we have $e \in Q$ iff $e \in S$ and
      $f_{S - e}(e) \geq \epsless \lambda$. Moreover, the events
      $[e \in S]$ and $[f_{S - e}(e) \geq \epsless \lambda]$ are
      independent because $S$ is an independent sample. Thus
      \begin{align*}
        \probof{e \in Q}          %
        =                         %
        \probof{e \in S} \probof{f_{S-e}(e) %
        \geq %
        \epsless \lambda}         %
        =                         %
        \delta
        \probof{f_{S - e}(e) \geq \epsless \lambda}
        \numberthis
      \end{align*}
      We have
      \begin{align*}
        \ef{Q}
        &\tago{\geq}
          \sum_{e \in \groundset}
          \probof{e \in Q}
          \evof{f_{Q - e}(e) \given e \in Q}         %
        \\
        &\tago{=}                           %
          \delta
          \sum_{e \in \groundset}     %
          \probof{f_{S-e}(e) \geq \epsless \lambda} %
          \evof{f_{Q-e}(e) \given e \in Q} %
        \\
        &\tago{\geq}
          \delta
          \sum_{e \in \groundset}
          \probof{f_{S - e}(e) \geq \epsless \lambda}
          \evof{f_{S - e}(e) \given e \in Q}
        \\
        &\tago{\geq}
          \delta \epsless \lambda
          \sum_{e \in \groundset}
          \probof{f_{S - e}(e) \geq \epsless \lambda}
        \\
        &\tago{\geq}
          \delta \epsless^2 \lambda \sizeof{\groundset}
          = \epsless^2 \lambda \evof{\sizeof{S}}
      \end{align*}
      by \tagr submodularity, \tagr substituting for
      $\probof{e \in Q}$ by equation \reflastequation above, \tagr
      submodularity and $Q - e \subseteq S - e$, \tagr definition of
      $Q$, and \tagr choice of $\delta$ per
      \refstep{gms-step-size-margin}.
    \item For each element $e$, $e \in S \setminus P$ iff $e \in S$
      and $e \in \spn{S - e}$. Moreover, the events
      $\bracketsof{e \in S}$ and $\bracketsof{e \in \spn{S - e}}$ are
      independent because $S$ is an independent sample. Thus
      \begin{align}
        \probof{e \in S \setminus P} = \delta \probof{e \in \spn{S-e}}
        \text{ for all } e \in \groundset.
      \end{align}
      We have
      \begin{align*}
        \evof{\sizeof{S \setminus P}} %
        &\tago{=} \sum_{e \in \groundset} \probof{e \in S \setminus P} %
          \tago{=}   %
          \delta \sum_{e \in \groundset} \probof{e \in \spn{S-e}}
        \\
        &\tago{\leq}                %
          \delta \sum_{e \in \groundset} \probof{e \in \spn{S}} %
          \tago{=}                                                      %
          \delta \evof{\sizeof{\spn{S}}}                               %
          \tago{\leq}                       %
          \delta \eps \sizeof{\groundset} = \eps\evof{\sizeof{S}}.
      \end{align*}
      by \tagr linearity of expectation, \tagr equation
      \reflastequation above, \tagr monotonicity of span, \tagr
      linearity of expectation, and \tagr choice of $\delta$.
    \end{enumerate}
  \item By maximality of $\delta$, either
    \refstep{gms-step-size-margin} or \refstep{gms-step-size-span} is
    tight. If \refstep{gms-step-size-margin} is tight, then
    \begin{align*}
      \evof{\sizeof{\setof{e \where f_S(e) \geq \epsless \lambda}}} %
      \leq                                                    %
      \epsless \sizeof{\groundset}.
    \end{align*}
    If \refstep{gms-step-size-span} is tight, then
    \begin{align*}
      \evof{\sizeof{\groundset \setminus \spn{S}}} %
      =
      \sizeof{\groundset} - \evof{\sizeof{\spn{S}}} %
      \leq                                          %
      \epsless \sizeof{\groundset}.
    \end{align*}
    Thus
    \begin{align*}
      \evof{\sizeof{\setof{e \notin \spn{S} \where f_S(e) \geq \epsless \lambda}}}
      &\leq %
        \min
        \begin{cases}
          \evof{\sizeof{\groundset \setminus \spn{S}}},\\
          \evof{\sizeof{\setof{e \where f_S(e) \geq \epsless \lambda}}}
        \end{cases}
      \\
      &\leq                       %
        \epsless \sizeof{\groundset}.
    \end{align*}
    as desired.
  \end{enumerate}
\end{proof}

\begin{remark}
  \reftheorem{greedy-sample} does not require $f$ to be nonnegative
  but does require $\lambda \geq 0$. \reftheorem{greedy-sample} holds
  for any independence system $\defmatroid$ equipped with a function
  $\spn : \groundsets \to \groundsets$ such that:
  \begin{mathproperties}
  \item If $S \subseteq T \subseteq \groundset$, then $\spn{S} \subseteq \spn{T}$.
  \item If $S \subseteq \groundset$ and $e \notin \spn{S - e}$ for all
    $e \in S$, then $S \in \independents$.
  \end{mathproperties}
  For matroids, this is the standard span function. For matroid
  intersections, this is given by the union of the individual span
  functions.
\end{remark}

\begin{remark}
  \labelremark{step-size-range} %
  Claim (b) in the proof implies that
  $\delta \sizeof{\groundset} \leq \bigO{\rank{\matroid}}$. Indeed,
  $\sizeof{S \setminus P} \geq \sizeof{S} - \rank{\matroid}$
  deterministically, and (b) asserts that
  $\evof{\sizeof{S \setminus P}} \leq \eps \evof{\sizeof{S}}$. On the
  other hand, we know that $\delta \geq \eps / n$, since otherwise
  $S = \emptyset$ with probability at least $1 - \eps$ and so
  \refstep{gms-step-size-margin} and \refstep{gms-step-size-span} are
  easily satisfied.
\end{remark}


\section{Iteration by greedy blocks}

\labelsection{greedy-blocks-apx}

\begin{figure}
  \begin{algorithm}{block-greedy}{$\defmatroid$,$\defrsubmodular$,$\eps$}
  \item $I, S \gets \emptyset$;
    $\lambda \gets \max_{e \in \groundset} f(e)$;
    $\lambda_0 \gets \frac{\eps \opt}{\rank{\matroid}}$ %
    \commentcode{or any
      $\lambda_0 \leq \frac{\eps \opt}{\rank{\matroid}}$}
  \item while
    $\lambda \geq \lambda_0$
    \begin{steps}
    \item \labelstep{bg-threshold-loop} while
      $\groundset' = \setof{e \in \groundset \where f_S(e) \geq
        \epsless \lambda}$ is not empty
      \begin{steps}
      \item $(I',S') \gets$ $\epsless$-greedy blocks w/r/t $f_S$ and
        $\parof{\matroid / S} \land \groundset'$
      \item $I \gets I \cup I'$, $S \gets S \cup S'$
      \end{steps}
    \item $\lambda \gets \epsless \lambda$
    \end{steps}
  \item return $(I,S)$
  \end{algorithm}
  \begin{implicitframed}
    \caption{An extension of the greedy algorithm to greedy blocks with
      a polylogarithmic number of iterations in expectation.}
  \end{implicitframed}
\end{figure}

We now iterate along greedy blocks, where each iteration is w/r/t the
residual of previously selected greedy blocks.  We start with empty
sets $I,S = \emptyset$. We repeatedly generate $\epsless$-greedy
blocks $(I',S')$ w/r/t the contracted matroid $\matroid / S$ and the
marginal values $f_S$, and then add $I'$ to $I$ and $S'$ to
$S$. Recursively selecting greedy blocks like so leads to
approximation factors resembling the standard greedy analysis
(\reftheorem{block-greedy} below).

The $\lambda$'s are generated by the standard technique of
thresholding, which is employed for the sake of efficiency. We
maintain a value $\lambda$ such that $f_S(e) \leq \lambda$ for all
$e \in \groundset$. We do not update $\lambda$ until
$f_S(e) \leq \epsless \lambda$ for all $e \in \groundset$, at which
point we replace $\lambda$ with $\epsless \lambda$. By part (ii) of
\reftheorem{greedy-sample}, we only expect to call
\algo{greedy-sample} $\bigO{\log {n} / \eps}$ times for a fixed value
of $\lambda$.  $\lambda$ can only decrease
$\bigO{\log{\rank{\matroid}} / \eps}$ times before the marginal values
w/r/t $S$ are small enough to be negligible. We expect to select
$\bigO{\frac{\log{n} \log{\rank{\matroid}}}{\eps^2}}$ greedy blocks
before the algorithm terminates.

First give a general analysis concerning finite sequences of greedy blocks.
\begin{lemma}
  \labellemma{greedy-blocks} Let $\defmatroid$ be a matroid and let
  $\defrsubmodular$ be a submodular function. Let
  $I_1,\dots, I_k \subseteq \groundset$ and
  $S_1,\dots S_k \subseteq \groundset$ be two sequences of random sets
  such that:
  \begin{mathproperties}
  \item For each $i$, $(I_i,S_i)$ is an $\epsless$-greedy block w/r/t
    $f_{\uS{i-1}}$ and $\matroid / \uS{i-1}$.
  \item $f_{\uS{k}}(T \setminus \spn{\uS{k}}) \leq \beta$ for all
    $T \in \independents$.
  \end{mathproperties}
  Then
  \begin{mathresults}
  \item
    \begin{math}
      \ef{\uI{k}} %
      \geq %
      \epsless \ef{\uS{k}}.
    \end{math}
  \item For a fixed independent set $T \in \independents$, there is a
    partition $T_{1},\dots,T_{k+1}$ of $T$ (depending on
    $S_1,\dots,S_k$) such that $T_i \cap \uS{i-1} = \emptyset$ for
    each $i$ and
    \begin{align*}
      \ef{\uI{k}}
      \geq                      %
      \epsless
      \sum_{i=1}^{k+1} \ef{T_i}{\uS{i-1}}
      - \epsless \beta.
    \end{align*}
  \end{mathresults}
\end{lemma}
\begin{proof}
  \begin{mathresults}
  \item We have
    \begin{align*}
      \evof{f(\uI{k})}          %
      &\tago{=}                        %
        \sum_{i=1}^k \ef{I_i}{\uI{i-1}} %
        \tago{\geq}                    %
        \sum_{i=1}^k \ef{I_i}{\uS{i-1}} %
      \\
      &\tago{\geq}                            %
        \epsless \sum_{i=1}^k  \evof{\sizeof{S_i} \max {0, \max_{e \notin \spn{\uS{u-1}}}
        \f{e}{\uS{i-1}}}}
      \\
      &\tago{\geq}                     %
        \epsless \sum_{i=1}^k \ef{S_i}{\uS{i-1}} %
        \tago{=}                                             %
        \epsless \ef{\uS{k}}
    \end{align*}
    by \tagr telescoping, \tagr submodularity and
    $\uI{i-1} \subseteq \uS{i-1}$, and \tagr assumption (a) that
    $(I_i,S_i)$ is a greedy block, \tagr submodularity, and \tagr
    telescoping.
  \item By \reflemma{brualdi-mapping-span}, we can partition
    $T \cap \spn{\uS{k}} = T_1 \cup \cdots \cup T_k$ such that for
    each $i$,
    \begin{align*}
      T_i \subseteq \groundset \setminus \spn{\uS{i-1}}
      \text{ and }
      \sizeof{T_i} \leq \rank{\uni{S}_i} - \rank{\uS{i-1}} \leq
      \sizeof{S_i}.
      \labelthisequation{matroid-brualdi-partition}
    \end{align*}
    Let $T_1,\dots,T_k$ be such a partition for each realization of
    $S_1,\dots,S_k$, and let $T_{k+1} = T \setminus \spn{\uS{k}}$. The
    partition $T_1,\dots,T_{k+1}$ of $T$ is randomized, and a
    deterministic function of $S_1,\dots,S_k$.  We have
    \begin{align*}
      \evof{f(\uI{k})}         %
      &\tago{=}                           %
        \sum_{i=1}^k \evof{f_{\uI{i-1}}(I_i)} %
        \tago{\geq}                                          %
        \sum_{i=1}^k
        \evof{
        f_{\uS{i-1}}(I_i)
        }
      \\
      &\tago{\geq}                       %
        \epsless                                      %
        \sum_{i=1}^k
        \evof{                      %
        \sizeof{S_i} \max{0, \max_{e \notin \spn{\uS{i-1}}} \f{e}{\uS{i-1}}}
        }                          %
      \\
      &\tago{\geq}
        \epsless
        \sum_{i=1}^k
        \evof{\sizeof{T_i} \max{0,\max_{e \notin \spn{\uS{i-1}}}
        \f{e}{\uS{i-1}}}}
      \\
      &\tago{\geq}
        \epsless
        \sum_{i=1}^k
        \ef{T_i}{\uS{i-1}}
        \tago{\geq}                    %
        \epsless \sum_{i=1}^k \ef{T_i}{\uS{k}} %
      \\
      &\tago{\geq}
        \epsless
        \sum_{i=1}^{k+1} \ef{T_i}{\uS{i-1}}
        -
        \epsless \beta
    \end{align*}
    by \tagr telescoping, \tagr submodularity and
    $\uI{i-1} \subseteq \uS{i-1}$, \tagr $(I_i,S_i)$ being a greedy
    block w/r/t $f_{\uS{i-1}}$ and $\matroid / \uS{i-1}$, and \tagr
    $\sizeof{T_i} \leq \sizeof{S_i}$, \tagr
    $T_i \subseteq \groundset \setminus \spn{\uS{i-1}}$, \tagr
    submodularity, and \tagr adding the (negative) term
    \begin{math}
      \epsless \ef{T_{k+1}}{\uS{k}} - \epsless \beta.
    \end{math}
  \end{mathresults}
\end{proof}

\begin{lemma}
  \labellemma{conditional-greedy-block}
  Let $\defrsubmodular$ be a normalized submodular function and let
  $(I,S)$ be a greedy block where $\probof{S \neq \emptyset} >
  0$. Then $(I,S)$ is a greedy block conditional on
  $S \neq \emptyset$.
\end{lemma}
\begin{proof}
  We have
  \begin{align*}
    \evof{\sizeof{S}} =  \probof{S \neq
    \emptyset} \evof{\sizeof{S} \given S \neq \emptyset},
  \end{align*}
  and since \tagr $S = \emptyset \implies I = \emptyset$ and \tagr $f$
  is normalized, we have
  \begin{align*}
    \ef{I}
    &\tago{=} \probof{S = \emptyset} f(\emptyset) + \probof{S \neq
      \emptyset} \evof{f(I) \given S \neq \emptyset}
    \\
    &\tago{=}                           %
      \probof{S \neq \emptyset} \evof{f(I) \given S \neq \emptyset}
  \end{align*}
  Thus
  \begin{align*}
    \evof{f(I) \given S \neq \emptyset}
    &=                           %
      \frac{\ef{I}}{\probof{S \neq \emptyset}}
      =                          %
      \epsless \frac{\evof{\sizeof{S}}}{\probof{S \neq \emptyset}}
      \max{0,\max_e f(e)}         %
    \\
    &=                           %
      \epsless \evof{\sizeof{S} \given S \neq \emptyset}
      \max{0, \max_e f(e)},       %
  \end{align*}
  as desired.
\end{proof}

\begin{theorem}
  \labeltheorem{block-greedy} Let $\defmatroid$ be a matroid and let
  $\defnnsubmodular$ be a nonnegative submodular function. With
  $\bigO{\frac{\log{n} \log{\rank{\matroid}}}{\eps^2}}$ calls to
  \algo{greedy-sample} in expectation, one can compute a randomized
  independent set $I \in \independents$ and a randomized sequence of
  $n = \sizeof{\groundset}$ sets $S_1,\dots,S_n$ such that
  \begin{enumerate}
  \item $I \subseteq \uS{n}$.
  \item $\ef{I} \geq \epsless \ef{\uS{n}}$.
  \item For any set $T \in \independents$, there is a partition
    $T_1,\dots,T_{n+1}$ of $T$ (depending on $S_1,\dots,S_n$) such
    that
    \begin{align*}
      \ef{I} \geq \epsless \sum_{i=1}^{n+1} \ef{T_{i}}{\uS{i-1}}.
    \end{align*}
  \end{enumerate}
\end{theorem}
\begin{proof}
  Consider \algo{block-greedy}. Each iteration of
  \refsubsteps{bg-threshold-loop} expects to decrease $\groundset'$ by a
  $\epsless$-multiplicative factor, and $\groundset'$ is at most
  $n$. So we expect at most $\bigO{\log{n} / \eps}$ iterations of
  \refsubsteps{bg-threshold-loop} per $\lambda$. Moreover, we have
  $\bigO{\log{\rank{\matroid}} / \eps}$ choices of $\lambda$. Thus we
  expect at most $\bigO{\log{n} \log{\rank{\matroid}}/\eps^2}$
  iterations total.

  Of the greedy samples computed by \algo{block-greedy}, let
  $(I_1,S_1), (I_2,S_2), \dots$ be the subsequence of nonempty greedy
  blocks with $S_i \neq \emptyset$. Since the $S_i$'s are disjoint,
  there are at most $n$ pairs in the sequence.  Appending empty greedy
  blocks if necessary, we have $n$ random pairs
  $(I_1,S_1),\dots,(I_n,S_n)$ such that:
  \begin{mathproperties}
  \item for each $i \in [n]$, by \reflemma{conditional-greedy-block},
    $(I_i,S_i)$ is an $\epsless$-greedy block w/r/t
    $f_{\spn{\uS{i-1}}}$ and $\matroid / \uS{i-1}$.
  \item deterministically we have
    \begin{math}
      \ef{e}{\uS{n}}               
      \leq
      \frac{c \eps \max_e
        \f{e}}{\poly{\rank{\matroid}}}
    \end{math}
    for any desired constant $c$.
  \end{mathproperties}
  The result now follows from applying \reftheorem{block-greedy} with
  $\beta \leq \eps \ef{\uI_n}$.
\end{proof}

To help interpret \reftheorem{block-greedy}, we note that the
sequential greedy algorithm gives the above guarantee
deterministically, with $I = \uS{n}$ and $\eps = 0$. That is, greedy
returns an independent set $I \in \independents$ and a sequence of
sets $S_1,\dots,S_k \subseteq \groundset$ (for $k = \rank{\matroid}$)
such that:
\begin{mathproperties}
\item $I \subseteq \uS{n}$.
\item For any set $T \in \independents$, there is a partition
  $T_1,\dots,T_{n+1}$ of $T$ (depending $S_1,\dots,S_{n+1}$) such
  that
  \begin{align*}
    \f{I} \geq \sum_{i=1}^{n+1} \f{T_i}{\uS{i-1}}.
  \end{align*}
\end{mathproperties}

\subsection{Beyond matroids}

It is well-known that the greedy analysis extends past matroids, to
intersections of matroids, and more generally, matchoids. Naturally,
greedy blocks extend to $p$-matchoids as well, obtaining the
appropriate approximation ratios inversely proportional to $p$. Here
we let $\rank{\matroid} = \max{\sizeof{I} \where I \in \independents}$
denote the maximum cardinality of any independent set.
\begin{theorem}
  \labeltheorem{matchoid-block-greedy} Let $\defmatroid$ be a
  $p$-matchoid, and let $\defrsubmodular$ be a submodular
  function. With $\bigO{\log{n} \log{\rank{\matroid}} / \eps}$ calls
  to \algo{greedy-sample} in expectation, one can compute a randomized
  independent set $I \in \independents$ and a sequence of
  $n = \sizeof{\groundset}$ sets $S_1,\dots,S_n$ such that
  \begin{mathresults}
  \item $I \subseteq S$ and $\ef{I} \geq \epsless \ef{S}$.
  \item For any $T \in \independents$, there is a partition
    $T_1,\dots,T_{n+1}$ of $T$ (depending on $S$) such that
    $T_i \cap \spn{\uS{i-1}} = \emptyset$ for each $i$ and
    \begin{align*}
      \ef{I} \geq \frac{1-\eps}{p} \sum_{i=1}^{n+1} \ef{T_i}{\uS{i-1}}.
    \end{align*}
  \end{mathresults}
\end{theorem}
\begin{proof}[Proof sketch]
  The adjustment from matroids to $p$-matchoids w/r/t greedy blocks is
  analogous to the adjustment from matroids to $p$-matchoids in the
  sequential setting, so we limit ourselves to a sketch. We have
  already observed that \algo{greedy-sample} holds for matchoids. The
  analysis of \reflemma{greedy-blocks} must also be adjusted. The
  adjustment is made to the partition of the competing set $T$ in line
  \refequation{matroid-brualdi-partition}, where instead the sets $T_i$
  are such that $\sizeof{T_i} \leq p \sizeof{S_i}$ for each
  $i$. Carrying this inequality through, the approximation factor
  drops by a factor of $p$.
\end{proof}

To help interpret \reftheorem{matchoid-block-greedy}, we note that the
sequential greedy algorithm can be interpreted as a sequence of exact
(i.e., $\eps = 0$) and deterministic greedy blocks. This gives the
above guarantee deterministically, with $I = \uS{n}$ and with
$\eps = 0$. That is, greedy returns an independent set
$I \in \independents$ and a sequence of sets
$S_1,\dots,S_k \subseteq \groundset$ (for $k = \sizeof{I}$) such that:
\begin{mathproperties}
\item $I \subseteq \uS{n}$.
\item For any set $T \in \independents$, there is a partition
  $T_1,\dots,T_{n+1}$ of $T$ (dependent on $S_1,\dots,S_{n+1}$) such
  that
  \begin{align*}
    \f{I} \geq \frac{1}{p}\sum_{i=1}^{n+1} \f{T_i}{\uS{i-1}}.
  \end{align*}
\end{mathproperties}

\subsection{Monotone approximation}

\reftheorem{block-greedy} and \reftheorem{matchoid-block-greedy}
directly lead to approximations in the monotone case essentially
matching the greedy algorithms.
\begin{theorem}
  Let $\defmatroid$ be a $p$-matchoid, and let $\defnnsubmodular$ be a
  monotonic submodular function. With
  $\bigO{\log{n} \log{\rank{\matroid}} / \eps}$ calls to
  \algo{greedy-sample} in expectation, one can compute a randomized
  independent set $I \in \independents$ such that
  \begin{align*}
    \ef{I} \geq \frac{1-\eps}{p+1} f(T)
  \end{align*}
  for any $T \in \independents$.
\end{theorem}

\begin{proof}
  We have
  \begin{align*}
    p \ef{I}
    &\tago{\geq} \epsless \sum_{i=1}^{n+1} \ef{T_i}{\uS{i-1}} %
      \geq                                                       %
      \ef{T}{\uS{n}}                                             %
    \\
    &\tago{\geq}                                                       %
      \f{T} - \f{\uS{n}}                                            %
      \tago{\geq}                                                       %
      \f{T} - \epsless \ef{I}
  \end{align*}
  by \tagr \reftheorem{matchoid-block-greedy}, \tagr mononoticity and
  \tagr \reftheorem{matchoid-block-greedy}, which essentially gives us
  the desired inequality up to rearranging of terms.
\end{proof}

\subsection{Nonnegative approximation}

Let $\defmatroid$ by a $p$-matchoid and $\defnnsubmodular$ be a
nonnegetive submodular function. If $f$ is not monotone, then the
approximation guarantees (i) and (ii) of \reftheorem{block-greedy} do
not immediately imply an actual approximation factor, as the sum on
the RHS of (ii) may be much smaller than $f(T)$.

We can remedy this and obtain an approximation as follows. We define a
new nonnegative submodular function $g$ by
\begin{align*}
  g(S) = \evof{f(S') \where S' \sim \beta S} = \mlf{\beta S}
\end{align*}
for some parameter $\beta \in (0,1)$ to be determined.  Let
$I \in \independents$ and $S_1,\dots,S_n$ by the sets returned by
\algo{block-greedy} applied to $g$.  Fix $T \in \independents$, and
partition $T = T_1 \dunion \cdots \dunion T_{n+1}$ per
\reftheorem{block-greedy} (ii). Let $J \sim \beta I$. For each $i$,
let $S_i' \sim \beta S_i$ and $T_i' \sim \beta T_i$. Observe that for
each $i$, $\parof{T_i' \cup \uS{i-1}'} \sim \parof{T_i \cup \uS{i-1}}$
because $T_i$ and $\uS{i-1}$ are disjoint.  We have
\begin{align*}
  p\ef{J} %
  &=                                        %
    p\evof{g(I)}                                     %
    \tago{\geq}                                            %
    \epsless
    \sum_{i=1}^{n+1} \evof{g_{\uS{i-1}}(T_i)}
  \\
  &\tago{=}                            %
    \epsless \sum_{i=1}^{n+1}            %
    \parof{
    \mlf{\beta \uS{i-1} + \beta T_i} - \mlf{\beta \uS{i-1}}
    }
  \\
  &\tago{\geq}                            %
    \epsless \sum_{i=1}^{n+1}            %
    \parof{
    \beta \mlf{\beta\uS{i-1} + T_i} + \parof{1 - \beta} \mlf{\beta \uS{i-1}}
    - \mlf{\beta \uS{i-1}}
    }\\
  &=                            %
    \epsless \beta \sum_{i=1}^n \parof{\mlf{T_i}{\beta \uS{i-1}}}
  \\
  &=                            %
    \epsless \beta \sum_{i=1}^n \ef{T_i}{\uS{i-1}'} %
    \tago{\geq}                         
    \epsless \beta \ef{T}{\uS{n}'} \\
  &=                         %
    \epsless \beta \ef{T \cup \uS{n'}} - \epsless \beta \eg{\uS{n}} %
  \\
  &\tago{\geq}                                                  %
    \epsless \beta \parof{1 - \beta}\f{T} - \epsless \beta \eg{\uS{n}}
  \\
  &\tago{\geq}
    \epsless \beta \parof{1-\beta} \f{T} - \beta \ef{J}
\end{align*}
by \tagr \reftheorem{matchoid-block-greedy} (ii), \tagr $T_i$ and
$\uS{i-1}$ are disjoint which implies that
\begin{math}
  g(T_i \cup \uS{i-1}) = \mlf{\beta \parof{T_i \cup \uS{i-1}}} =
  \mlf{\beta T_i + \beta \uS{i-1}},
\end{math} \tagr monotonic concavity (between $\beta \uS{i-1}$ and
$\beta \uS{i-1} + T_{i}$), \tagr submodularity and
$\uS{i}' \subseteq \uS{n}'$, \tagr submodularity and that
$T_1,\dots,T_{n+1}$ partition $T$, \tagr observing that
$\probof{e \in \uS{n}' \setminus T} \leq \beta$ for all $e$ and
applying \reflemma{lovasz-floor}, and \tagr
\reftheorem{matchoid-block-greedy} (i).

Rearranging, we have
\begin{align*}
  \ef{J} \geq \epsless \frac{\beta(1-\beta)}{\beta + p} \f{T}
\end{align*}
For $p = 1$, and $\beta = \sqrt{2} - 1$, we have
\begin{align*}
  \ef{J} &\geq \epsless \parof{3 - 2 \sqrt{2}} \f{T}
  &\text{(which is $\approx .17 \f{T}$)}.
\end{align*}
In general, by setting $\beta = \sqrt{p(p+1)} - p$, which gives
\begin{align*}
  \ef{J} \geq \epsless \parof{2 p - 2\sqrt{p(p+1)} + 1} \f{T}.
\end{align*}
Note that $J \subseteq I \in \independents$ implies $I \in \independents$.

\begin{corollary}
  Let $\defmatroid$ be a $p$-matchoid and $\defnnsubmodular$ a
  nonnegative submodular function. For sufficiently small $\eps > 0$,
  with $\bigO{\log{n} \log{\rank{\matroid}} / \eps^2}$ adaptive rounds
  in expectation, one can compute an independent set $J$ such that
  \begin{align*}
    \ef{I} \geq \epsless \parof{2 p + 1 - 2 \sqrt{p(p+1)}} \opt %
    =                                                           %
    \epsless \prac{1 + o(1)}{4 (p+1)} \opt.
  \end{align*}
\end{corollary}

             %
%
\section{Multilinear amplification}
\labelsection{multilinear-amplification}

\providecommand{\epsB}{\eps'}   %
\newcommand{\indgs}{\independents \times \naturalnumbers}
\providecommand{\epsBless}{\parof{1 - \epsB}}
\providecommand{\epsBmore}{\parof{1 + \epsB}}
\providecommand{\apxBless}{\parof{1 - \bigO{\epsB}}}
\newcommand{\fpi}{\f[f^{\pi}]}
\newcommand{\efpi}{\ef[\fpi]}

By combining greedy matroid sampling of \refsection{greedy-sample}
with the iterative analysis of greedy blocks in
\refsection{greedy-blocks-apx}, we can produce independent sets that
resemble those of the greedy algorithm in performance. For monotone
functions, this implies a near-2 approximation, short of the desired
near-$\parof{1 - 1/e}$ approximation factor achieved originally by the
continuous greedy algorithm in \cite{ccpv}.  \citet{bv} showed how to
amplify a near-2 approximation into a near-$\parof{1 - 1/e}$
approximate fractional solution, as a convex combination of
$\bigO{\reps}$ near-2 approximations w/r/t submodular functions
induced by the multilinear relaxation. We call this process
``multilinear amplification''.

We want to highlight that multilinear amplification applies to any set
system, and for any approximation guarantee for the oracle, where the
amplified approximation factor varies with the approximation ratio of
the oracle.

The primary caveat with multilinear amplification is that the
fractional solution still needs to be rounded. Many constraints can be
rounded with some bounded loss, but in many cases the loss offsets the
gains from the amplification.\footnote{Perhaps this is why
  \cite{bv} does not explicitly explore multilinear amplification
  in its full generality.}  An important exception is matroids, which
can be rounded without loss. For other constraints, the amplification
step may still prove valuable if the rounding schemes improve in the
future.

It is not clear if general matroids can be rounded with low
depth. Some explicit matroids, such as partition matroids, can be
handled fairly easily and directly. Rounding algorithms for general
matroids, such as swap rounding, appeal to properties (such as
Brualdi's strong exchange) that are inherently sequential.

The amplification process still gives an improved approximation ratio
to the \emph{value} of the optimization problem, since there is no
integrality gap for matroids. Moreover, rounding to a matroid in
parallel may be less important than solving the optimization problem
in parallel, because the $\bigO{1/\eps}$ independent sets produced by
the amplification may be much smaller than the original input, and
existing rounding schemes are oblivious do not require $f$. We think
that parallelizable rounding schemes for matroids is an interesting
open question.

\subsection{Monotone multilinear amplification}

\begin{figure}
  \begin{algorithm}{monotone-ML-amp}{$\defmatroid$,
      $\defnnsubmodular$, $\eps > 0$}
  \item $x \gets \zeroes$
  \item repeat $\ell = \bigO{\reps}$ times
    \begin{steps}
    \item define
      \begin{math}
        g(S) = \mlf{S / \ell}{x} %
      \end{math}                 %
    \item invoke oracle to get $S \in \independents$ s.t.\
      \begin{center}
        \begin{math}
          \eg{S} \geq \epsless \alpha \eg{T}{S}
        \end{math}
        for all $T \in \independents$
      \end{center}
    \item $x \gets x + \frac{S}{\ell}$
    \end{steps}
  \item return $x$
  \end{algorithm}
  \begin{implicitframed}
    \caption{A procedure based on \citep{bv} amplifying an
      $\alpha$-approximation algorithm for monotone submodular
      maximization to a fractional nearly
      $\parof{1-e^{-\alpha / 1- \alpha}}$-approximation algorithm via
      the multilinear extension. \labelfigure{ml-amp}}
  \end{implicitframed}
\end{figure}

The monotone case is essentially given by \citet{bv} and sketched
in \reffigure{ml-amp}.

\begin{theorem}
  \labeltheorem{multilinear-amplification}
  Let $\defmatroid$ be a set system, and let
  $f: \groundsets \to \nnreals$ be a monotone submodular
  function. Suppose one has access to an oracle that, given a
  submodular function $g$, returns a (possibly randomized) set
  $S \in \independents$ such that
  \begin{center}
    \begin{math}
      \evof{g(S)} \geq \alpha \evof{g_S(T)}
    \end{math}
    for all $T \in \independents$,
  \end{center}
  for some $\alpha \in (0,1]$.  With $\bigO{\reps}$ calls to this
  oracle, one can compute a convex combination of $\bigO{\reps}$
  independent sets $x \in \convex{\independents}$ such that
  \begin{center}
    \begin{math}
      F(x) \geq \epsless \parof{1 - e^{-\alpha}} f(T)
    \end{math}
    for all $T \in \independents$,
  \end{center}
  where $F$ is the multilinear extension of $f$.
\end{theorem}

\begin{proof}
  The proof is similar to that of \citet{bv}.  Compared to \citep{bv},
  we make the minor adjustments of taking appropriate expected values
  due to the randomized oracle, and to handle the variable
  approximation ratio $\alpha$. We also take the oracle guarantee as
  something of a black box independent of the amplification, whereas
  the greedy analysis and amplification framework are somewhat
  intertwined in the proof of \citep{bv}. The key idea, that
  $\bigO{1/\eps}$ bases lead to a $\parof{1- 1/e}$-approximation, is
  due to \citep{bv}.

  Fix an independent set $T \in \independents$.  For each $i$, let
  $S_i \in \independents$ be the set selected in the $i$th round. Let
  $x_i = \frac{1}{\ell} \sum_{j \leq i} S_j \in \groundcube$ denote
  the vector accumulated after $i$ rounds.  Let
  $\delta_i = \max{0,f(T) - F(x_i)}$.  We assume each $S_i$ and $T$
  are disjoint by duplicating each element (as a thought experiment).
  We have
  \begin{align*}
    \evof{\delta_{i-1} - \delta_i}
    &= \evof{\mlf{x_i} - \mlf{x_{i-1}}}
      =                         %
      \evof{g(S_i)}
    \\
    &\tago{\geq}
      \alpha \evof{g_{S_i}(T)}
      =                          %
      \alpha \evof{g(S_i \cup T) - g(S_i)}
    \\
    &= \alpha
      \evof{
      \mlf{x_{i} + T / \ell} - \mlf{x_{i-1}}
      - \parof{\mlf{x_i} - \mlf{x_{i-1}}}}
    \\
    &=                          %
      \alpha
      \evof{\mlf{x_i + T / \ell} - \mlf{x_i}}
      =                          %
      \alpha\evof{
      \mlf{T / \ell}{x_i}
      }
    \\
    &\tago{\geq}
      \alpha \evof{\parof{\frac{1}{\ell} \mlf{x_i + T} + \parof{1 - \frac{1}{\ell}}\mlf{x_i}} - \mlf{x_i}}
    \\
    &=                       %
      \frac{\alpha}{\ell}\evof{\mlf{T}{x_{i}}}
      \tago{\geq}                          %
      \frac{\alpha}{\ell}\parof{f(T) - \evof{\mlf{x_i}}}
      =                          %
      \frac{\alpha}{\ell}\evof{\delta_i}.
  \end{align*}
  by \tagr the oracle guarantee, \tagr monotonic concavity (between
  $x_i$ and $x_i + T$) and \tagr monotonicity.  Rearranging, we
  have
  \begin{math}
    \evof{\delta_i} \leq \parof{1 + \frac{\alpha}{\ell}}^{-1}
    \evof{\delta_{i-1}},
  \end{math}
  which upon unrolling leads to
  \begin{align*}
    f(T) - \evof{F(x_{\ell})}
    &= \evof{\delta_{\ell}} %
      \leq \parof{1 + \frac{\alpha}{\ell}}^{-\ell} \delta_0 %
      =                           %
      \parof{1 + \frac{\alpha}{\ell}}^{-\ell} f(T) %
    \\
    &\leq                        %
      e^{-\parof{1 - \bigO{\alpha/\ell}}\alpha} f(T) %
      \leq                                           %
      \epsless e^{-\alpha} f(T)
  \end{align*}
  for $\ell = \bigO{\reps}$, which is the desired inequality up to
  rearrangement of terms.
\end{proof}

\subsection{Nonnegative multilinear amplification}

\begin{figure}
  \begin{algorithm}{nonnegative-ML-amp}{$\defmatroid$,
      $\defnnsubmodular$, $\eps > 0$, $\alpha \in (0,1]$}
  \item $x \gets \zeroes$
  \item for $i = 1$ up to $\ell = \bigO{\reps}$
    \begin{steps}
    \item define $g(S) = \evof{f_{\uJ{i-1}}(S')}$ where $S' \sim S /
      \ell$ and $J_{j} \sim \alpha I_j / \ell$ for each $j$
    \item invoke oracle to get $I_i \in \independents$ satisfying
      conditions (a) and (b) of
      \reftheorem{nn-multilinear-amplification}
    \end{steps}
  \item return $(I_1,\dots,I_\ell)$
  \end{algorithm}
  \begin{implicitframed}
    \caption{A procedure amplifying greedy-type approximation algorithms
      for nonegative submodular functions via the multilinear
      extension. \labelfigure{nn-ml-amp-2}}
  \end{implicitframed}
\end{figure}

In this section, we consider another multilinear amplification scheme
that attains weaker bounds for the more general class of nonnegative
submodular functions.  The nonnegative multilinear amplification
scheme is to measured greedy what the monotone multilinear
amplification scheme is to continuous greedy, as explained below in
\refremark{measured-greedy=amplification}.

\begin{theorem}
  \labeltheorem{nn-multilinear-amplification} Let $\defmatroid$ be a
  set system, and let $\defnnsubmodular$ be a nonnegative submodular
  function. Suppose one has access to an oracle that, given a
  submodular function $g$, computes a (possibly randomized) set
  $I \subseteq \groundset$ and for which there exists a sequence of
  disjoint sets $S_1, S_2, \dots, S_k \subseteq \groundset$ such that
  \begin{mathproperties}
  \item
    \begin{math}
      \evof{g(I)} \geq \epsless \evof{g(\uS{k})}
    \end{math}
  \item For any $T \in \independents$, there exists a partition
    $T = T_1 \cup \cdots \cup T_{k+1}$ (depending on $S_1,\dots,S_k$)
    such that
    \begin{center}
      \begin{math}
        T_j \cap \uS{j-1} = \emptyset
      \end{math}
      for each $j$ and
      \begin{math}
        \evof{g(I)} \geq \alpha \sum_j \evof{g_{\uS{j-1}}(T_j)}.
      \end{math}
    \end{center}
    for some $\alpha \in [0,1]$.
  \end{mathproperties}
  With $\ell = \bigO{\reps}$ calls to this oracle, one can compute
  $\ell$ independent sets $I_1,\dots,I_{\ell} \in \independents$ such
  that if $J_i \sim \alpha I_i / \ell$ independently for each
  $i \in [\ell]$, then
  \begin{align*}
    \ef{J_1 \cup \cdots \cup J_{\ell}} \geq \epsless \alpha e^{-\alpha} f(T)
  \end{align*}
  for all $T \in \independents$.
\end{theorem}

\begin{proof}
  Fix an independent set $T \in \independents$. For each $i$, let
  $I_i \in \independents$ be the independent set selected in the $i$th
  iteration. Let $J_i \sim I_i / \ell$ independently for each $i$. We
  claim, for each $i$, that
  \begin{align*}
    \ef{\uJ{i}}{\uJ{i-1}}       %
    &\geq  %
      \frac{\alpha}{\ell} \parof{ %
      \prac{\parof{1 - \alpha/\ell}^i}{1
      + \eps \alpha / \ell}  f(T) - \ef{\uJ{i-1}}
      }
      \labelthisequation{nn-amp-recurrence}
    \\
    \bigg(
    \text{which is }
    &\approx \frac{\alpha}{\ell} \parof{e^{-\alpha i/\ell} \opt -
      \ef{\uJ{i-1}}}
      \bigg).
  \end{align*}
  Fix $i$. Let $S_{1}, \dots S_k \subseteq \groundset$ be the family
  of sets satisfying (a) and (b) w/r/t $I_i$, and let
  $T = T_1 \dunion \cdots \dunion T_{k+1}$ be the corresponding
  partition of $T$ satisfying (b). Let $S_{j}' \sim \alpha \eps S_{j}$
  for each $i$ and let $T_j' \sim \alpha T_j / \ell$ for each $j$.
  \begin{align*}
    \ef{J_{i}}{\uJ{i-1}}
        &= g(I_i)                    %
          \tago{\geq} \alpha\sum_{j=1}^k \evof{g_{\uS{j-1}}(T_k)}
          =                     %
          \alpha \sum_{j=1}^k \evof{g(T_k \cup \uS{j-1})} - \evof{g(T_k)}
    \\
    &\tago{=}                          %
      \alpha \sum_{j=1}^k \ef{\uS{j-1}' \cup T'_k}{\uJ{i-1}} -
      \ef{\uS{j-1}'}{\uJ{i-1}}
    \\
    &=                          %
      \alpha \sum_{j=1}^k       %
      \ef{T'_k \cup \uJ{i-1} \cup \uS{j-1}'} %
      - \ef{\uJ{i-1} \cup \uS{j-1}'} %
    \\
    &                       %
      \begin{aligned}
        \tago{\geq} \alpha \sum_{j=1}^k &\frac{1}{\ell} \ef{T_k
          \cup \uJ{i-1} \cup \uS{j-1}'} + \parof{1 - \frac{1}{\ell}}
        \ef{\uJ{i-1} \cup \uS{j-1}'}
        \\
        &- \ef{\uJ{i-1} \cup \uS{j-1}'}
      \end{aligned}
    \\
    &=                          %
      \frac{\alpha}{\ell} \sum_{j=1}^k \ef{T_k}{\uJ{i-1} \cup
      \uS{j-1}'}
      \tago{\geq}
      \frac{\alpha}{\ell} \sum_{j=1}^k \ef{T_k}{\uJ{i-1}
      \cup \uS{k}'}
    \\
    &
      \tago{\geq}                                      %
      \frac{\alpha}{\ell}\ef{T}{\uJ{i-1} \cup \uS{k}'} %
    =                           %
      \frac{\alpha}{\ell} \parof{\ef{\uS{k}' \cup \uJ{i-1} \cup T} -
      \ef{\uS{k}' \cup \uJ{i-1}}}
      \labelthisequation{nn-amp-1}
  \end{align*}
  by \tagr the oracle guarantee, \tagr definition of $g$ and that
  $\uS{j-1}$ and $T_k$ are disjoint (!), \tagr monotonic concavity,
  \tagr submodularity and $\uS{j-1}' \subseteq \uS{k}'$, and \tagr
  submodularity and that $\setof{T_1,\dots,T_{k+1}}$ is a partition of
  $T$.  Now, we have
  \begin{align*}
    \ef{T \cup \uS{k}' \cup \uJ{i-1}} \geq \parof{1 -
    \frac{\alpha}{\ell}}^i f(T)
    \labelthisequation{nn-amp-2}
  \end{align*}
  by \reflemma{lovasz-floor} because
  \begin{center}
    \begin{math}
      \probof{e \in \uS{k}' \cup \uJ{i-1}} \leq 1 - \parof{1 -
        \frac{\alpha}{\ell}}^i
    \end{math}
    for any $e \in \groundset$.
  \end{center}
  We also have
  \begin{align*}
    \ef{S_i' \cup \uJ{i-1}}     %
    &=                           %
      g(S_i) + \ef{\uJ{i-1}}      %
      \tago{\leq}                        %
      \epsmore g(J_i) + \ef{\uJ{i-1}} %
    \\
    &=                               %
      \ef{\uJ{i}} + \eps \ef{J_i}{\uJ{i-1}}
      \labelthisequation{nn-amp-3}
  \end{align*}
  by the oracle guarantee (a). Plugging into inequalities
  \refequation{nn-amp-2} and \refequation{nn-amp-3} into line
  \refequation{nn-amp-1}, we have
  \begin{align*}
    \ef{J_i}{\uJ{i-1}}          %
    \geq                       %
    \cdots                     %
    \geq                       %
    \refequation{nn-amp-1}     %
    \geq                       %
    \frac{\alpha}{\ell}\parof{\parof{1 - \frac{\alpha}{\ell}}^i f(T) -
    \ef{\uJ{i}} - \eps \ef{\uJ{i-1}}},
  \end{align*}
  which gives \refequation{nn-amp-recurrence} up to rearrangement of
  terms (and noting that $\ef{\uJ{i}} \geq 0$).

  To solve the recurrence \refequation{nn-amp-recurrence}, let use denote
  \begin{align*}
    x_i =\frac{\parof{1 + \alpha/\ell}\ef{\uJ{i}}}{f(T)}.
  \end{align*}
  Dividing \refequation{nn-amp-recurrence} by
  $f(T) / \parof{1 + \alpha \ell}$ and substituting in gives
  \begin{align*}
    x_i - x_{i-1} \geq \frac{\alpha}{\ell} \parof{\parof{1 -
    \frac{\alpha}{\ell}}^{i} - x_{i-1}};
  \end{align*}
  rearranging, we have
  \begin{align*}
    x_i \geq \parof{1 - \frac{\alpha}{\ell}}x_{i-1} +
    \frac{\alpha}{\ell} \parof{1 - \frac{\alpha}{\ell}}^{i}
    \numberthis.
  \end{align*}
  We claim by induction on $i$ that
  \begin{align*}
    x_i \geq \frac{i \alpha}{\ell} \parof{1 - \frac{\alpha}{\ell}}^{i}.
  \end{align*}
  We have $x_0 = 0$. For larger indices, plugging into
  \reflastequation inductively, we have
  \begin{align*}
    x_{i+1}                     %
    &\geq                       %
      \parof{1 - \frac{\alpha}{\ell}} %
      \parof{\frac{i \alpha}{\ell} \parof{1 - \frac{\alpha}{\ell}}^i} %
      + \frac{\alpha}{\ell}  \parof{1 - \frac{\alpha}{\ell}}^{i+1}
          =
          \frac{(i+1) \alpha}{\ell} \parof{1 - \frac{\alpha}{\ell}}^{i+1},
  \end{align*}
  as desired.

  Now, for sufficiently large $\ell = \bigO{\reps}$, we have
  \begin{align*}
    \ef{\uJ{\ell}} \geq \parof{1 - \bigO{\frac{1}{\ell}}} f(T) x_\ell
    \geq                        %
    \parof{1 - \bigO{\frac{1}{\ell}}} \alpha \parof{1 - \frac{\alpha}{\ell}}^\ell
    \geq                        %
    \epsless \alpha e^{-\alpha},
  \end{align*}
  as desired.
\end{proof}

\begin{remark}
  \labelremark{measured-greedy=amplification} The proof gives some
  intuition for the ``measured greedy'' algorithm. The union of random
  subsets $\uJ{\ell}$ form an independent sample with
  $\uJ{\ell} \sim x$ for the vector $x$ defined by
  \begin{align*}
    x_e = \probof{e \in \uJ{\ell}}
    =                           %
    1 - \prod_{i=1}^{\ell} \probof{e \notin J_{i}}
    =                           %
    1 - \parof{1 - \frac{1}{\ell}}^{k_e}
    \approx
    1 - e^{-k_e / \ell}
    \text{ for }
    k_e = \sizeof{\setof{i \where e \in I_i}}.
  \end{align*}
  The margins $x$ are essentially of the form constructed by measured
  greedy, and grow from one iteration to the next in a similarly
  nonlinear, ``measured'' fashion.
\end{remark}

\begin{remark}
  \labelremark{nn-oracle-complexity} One motivation for \citet{bv} was
  to improve the oracle complexity of maximizing a monotone submodular
  function subject to a matroid constraint in the sequential
  setting. The nonnegative amplification scheme appears to offer some
  improvement in the oracle complexity of maximizing a generally
  nonnegative submodular function subject to a matroid constraint in
  the sequential setting.  It is cleaner to analyze the sequential
  setting so we plan to address this in a separate writeup.
\end{remark}

 %
%
\section{Implementation~details regarding estimation and~sampling}

\labelsection{estimation}

\providecommand{\eh}{\ef[h]} %

\newcommand{\sumX}{\overline{X}\optsub}%
\newcommand{\optgs}{\delta^{\star}} %
\newcommand{\optS}{S^{\star}}       %
\newcommand{\opti}{{i^{\star}}}     %
\newcommand{\cepsless}{\parof{1 - c \eps}} %
\newcommand{\cepsmore}{\parof{1 + c \eps}} %
\newcommand{\cmore}{\parof{1 + c}}         %
\newcommand{\cless}{\parof{1 - c}}
\newcommand{\apxcless}{\parof{1 - \bigO{c}}} %
\newcommand{\apxcmore}{\parof{1 + \bigO{c}}} %
\newcommand{\apxcpm}{\parof{1 \pm \bigO{c}}}

Our algorithms were described as if one can evaluate certain expected
values exactly without significant overhead. Here we show that these
randomized steps can be implemented with sufficient accuracy and
confidence without incurring any increase in adaptivity and within a
reasonable number of oracle calls\footnote{In general, we are not
  trying to optimize the oracle complexity.}. The probabilistic
analysis is tedious and similar analyses occur in the literature
(particularly w/r/t estimating the multilinear relaxation). We
postponed these technical details to this point because we feel they
obscure the combinatorial character of the algorithms.

We employ the following standard Chernoff inequality that allows for
both multiplicative and additive error.

\begin{lemma}[Chernoff]
  Let $X_1,\dots,X_n \in [0,1]$ be independent random variables,
  $\mu = \evof{\sum_i X_i}$ their expected value, $\eps > 0$
  sufficiently small, and $\gamma > 0$. Then
  \begin{align*}
    \probof{\absvof{X - \mu} \geq \eps \mu + \gamma}   %
    \leq                                               %
    c \exp{- d \eps \gamma}
  \end{align*}
  where $c,d > 0$ are universal constants.
\end{lemma}

Let $k = \max_{I \in \independents} \sizeof{I}$ be the maximum
cardinality of any independent set. We can assume that
$k \geq \bigO{\log n}{\eps^2}$, since otherwise we could have run the
sequential greedy algorithm.

\subsection{Approximating the greedy step size}

In this section, we show how to find a sufficiently good value of
$\delta$ in \algo{greedy-sample}. Our goal is to prove the following.

\begin{lemma}
  \labellemma{greedy-probability-search} Let $\defmatroid$ be a
  matroid or $p$-matchoid for fixed $p$, let $\defrsubmodular$ be a
  submodular function.  Let $\lambda \geq 0$ such that
  $\epsless \lambda \leq \f{e} \leq \lambda$ for all
  $e \in \groundset$.  With probability $\geq 1- 1/\poly{n}$, with
  $\bigO{1}$ adaptive rounds and $\bigO{k \log{n}/\eps^3}$ oracle
  queries to $f$, one can compute $\delta > 0$ such that for
  $S \sim \delta \groundset$,
  \begin{enumerate}
  \item
    \begin{math}
      \evof{\spn{S}} \leq \eps \sizeof{\groundset}.
    \end{math}
  \item \begin{math} %
      \evof{\sizeof{e \where \f{e}{S} \leq \epsless \lambda}} %
      \leq                                                    %
      \eps \sizeof{\groundset}.                               %
    \end{math}
  \item Either
    \begin{math}
      \evof{\sizeof{\spn{S}}} %
      \geq %
      \frac{\eps}{2} \sizeof{\groundset} %
    \end{math}
    or
    \begin{math}
      \evof{\sizeof{e \where \f{e}{S} \leq \epsless \lambda}} %
      \geq %
      \frac{\eps}{2} \sizeof{\groundset}
    \end{math}
    (or both)
  \end{enumerate}
\end{lemma}

\begin{remark}
  \reflemma{greedy-probability-search} does not require $f$ to be
  nonnegative.
\end{remark}

We are concerned with two quantities for fixed $\delta$, (a)
\begin{math}
  \evof{\sizeof{e \where \f{S}{e} \leq \epsless \lambda}}
\end{math}
and (b)
\begin{math}
  \evof{\sizeof{\spn{S}}},
\end{math}
for an independent sample $S \sim \delta \groundset$. Ideally we want
to find the largest value $\delta$ such that both of these quantities
are at most $\eps \sizeof{\groundset}$, but it suffices to find a
value $\delta$ such that both quantities are at most
$\bigO{\eps \sizeof{\groundset}}$ and at least one of the two
quantities is at least $\bigOmega{\eps \sizeof{\groundset}}$. In
particular, we are allowed a (substantive) additive error of
$\bigOmega{\eps \sizeof{\groundset}}$.

\subsubsection{Concentration}

We first show that the expected values are concentrated. We will
eventually use the following two theorems with $\gamma \approx \eps n$
and $m = \bigO{\log{n} / \eps^2}$.

\begin{lemma}
  \labellemma{margin-concentration} Let $\eps, \gamma > 0$ and
  $m \in \naturalnumbers$.  Let $\defrsubmodular$ be a real-valued set
  function. Consider a fixed distribution over sets
  $S \subseteq \groundset$. With $\bigO{m}$ independent samples from
  this distribution and $\bigO{m}$ oracle queries to $\f$, one can
  compute a value $X \in \reals$ such that
  \begin{center}
    for
    $\mu = \evof{\sizeof{e \where \f{e}{S} \geq \epsless \lambda}}$,
    \begin{math}
      \probof{                    %
        \absvof{X - \mu}            %
        \geq                        %
        \eps \mu                    %
        +                           %
        \gamma                      %
      }                           %
      \leq                        %
      c \exp{- d \eps \gamma m / n}   %
      ,
    \end{math}
  \end{center}
  where $c, d > 0$ are universal constants.
\end{lemma}
\begin{proof}
  We define $m$ independent indicator variables as follows.  For
  $i \in [m]$, let $S_i$ be an independent sample from the
  distribution, and let $e_i \in \groundset$ be sampled uniformly at
  random. We define $X_i \in \setof{0,1}$ to be
  \begin{align*}
    X_i = \begin{cases}
      1 &\text{if } \f{e_i}{S_i} \geq \epsless \lambda,\\
      0 &\text{otherwise.}
    \end{cases}
  \end{align*}
  Then $\evof{X_i} = \mu / n$. Let $\sumX{m} = \sum_{i=1}^m X_i$ whe
  have $\evof{\sumX{m}} = m \frac{\mu}{n}$. By the Chernoff
  inequality, we have
  \begin{align*}
    \probof{\absvof{\frac{n}{m}\sumX{m} - \mu} \geq \eps \mu + \gamma} %
    &=
      \probof{\absvof{\sumX{m} - \evof{\sumX{m}}} \geq \eps \evof{\sumX{m}} +
      \frac{\gamma m}{n}}                     %
    \\
    &\leq                        %
      c e^{- d \eps \gamma m / n}
  \end{align*}
  for universal constants $c,d$.
\end{proof}

\begin{lemma}
  \labellemma{span-concentration} Let $\eps, \gamma > 0$ and
  $m \in \naturalnumbers$. Let $\defmatroid$ be a $p$-matchoid for
  fixed $p$. Consider a fixed distribution over
  $S \subseteq \groundset$. With $\bigO{m}$ independent samples from
  this distribution, one can compute a value $X \in \reals$ such that
  \begin{center}
    for $\mu = \evof{\sizeof{\spn{S}}}$,
    \begin{math}
      \probof{                    %
        \absvof{X - \mu}            %
        \geq                        %
        \eps \mu                    %
        +                           %
        \gamma                      %
      }                           %
      \leq                        %
      c \exp{- d \eps \gamma m}   %
      ,
    \end{math}
  \end{center}
  and $c, d > 0$ are universal constants.
\end{lemma}
\begin{proof}
  The cardinality of the span, $f(S) = \sizeof{\spn{S}}$, is a
  nonnegative submodular function. The result follows from
  \reflemma{margin-concentration}.
\end{proof}

\subsubsection{Precision and sensitivity}

Next, we analyze the sensitivity of the expected values
$\evof{\sizeof{\spn{S}}}$ and
$\evof{\sizeof{e \where \f{S}{e} \leq \epsless \lambda}}$ for
$S \sim \delta \groundset$ to slight changes in $\delta$. In
particular, we want to show that sufficiently small changes to
$\delta$ change the quantities only negligibly. This will allow us to
discretize the search for $\delta$ and bound the number of candidate
values.  We first prove a slightly more generic lemma, which will then
be applied to both quantities of interest.

\begin{lemma}
  \labellemma{monotone-estimate-additive} Let $\delta > \delta' > 0$.
  Let $S \sim \delta \groundset$, and let
  $S' \sim \delta' \groundset$. For any bounded, nonnegative, and
  monotonically increasing function $g: \groundsets \to [0,B]$,
  \begin{align*}
    \eg{S'}
    \leq \eg{S} \leq \eg{S'} + \parof{\delta - \delta'} \sizeof{\groundset} B.
  \end{align*}
\end{lemma}
\begin{proof}
  Couple $S$ and $S'$ so that $S \subseteq S'$ and
  $\probof{S \neq S'} \leq \parof{\delta - \delta'}
  \sizeof{\groundset}$. (e.g., for each $e \in \groundset$, we draw
  $\tau \in [0,1]$ uniformly at random, and include $e \in S$ if
  $\tau \leq \delta$ and include $e \in S'$ if $\tau \leq \delta'$.)
  Since $S' \subseteq S$ always and $g$ is monotonically increasing,
  \begin{align*}
    \eg{S'} \leq \eg{S}.
  \end{align*}
  In the opposite direction, we first have
  \begin{align*}
    \eg{S} = \evof{g(S) \given S = S'} \probof{S = S'} %
    +                                                  %
    \evof{g(S) \given S \neq S'} \probof{S \neq S'}.
    \numberthis
  \end{align*}
  For the first term, we have
  \begin{align*}
    \eg{S'}                                                    %
    &= \evof{g(S') \given S = S'} \probof{S = S'} %
      +                                                    %
      \evof{g(S') \given S \neq S'} \probof{S \neq S'}     %
    \\
    &\tago{\geq}                                                 %
      \evof{g(S) \given S = S'} \probof{S = S'}
  \end{align*}
  by \tagr nonnegativity of $g$.  For the second term, we have
  $\evof{g(S) \given S = S'}\leq B$ and
  \begin{math}
    \probof{S \neq S'}
    \leq                        %
    \parof{\delta - \delta'} \sizeof{\groundset} %
  \end{math}
  by the coupling. Plugging into \reflastequation, we have
  \begin{align*}
    \eg{S} \leq \eg{S'} + \parof{\delta - \delta'} \sizeof{\groundset} B.
  \end{align*}
\end{proof}

We now apply \reflemma{monotone-estimate-additive} to the quantities
of interest.
\begin{lemma}
  \labellemma{sensitivity} Let $\delta > \delta' > 0$. Let
  $\defrsubmodular$ be a submodular set function and $\defmatroid$ a
  matroid or $p$-matchoid for fixed $p$. Let
  $S \sim \delta \groundset$, and let $S' \sim \delta'
  \groundset$. Then
  \begin{align*}
    \evof{\sizeof{\spn{S'}}}    %
    \leq                        %
    \evof{\sizeof{\spn{S}}}     %
    \leq                        %
    \evof{\sizeof{\spn{S'}}}    %
    +                           %
    \parof{\delta - \delta'} \sizeof{\groundset}^2 %
  \end{align*}
  and
  \begin{align*}
    \evof{\sizeof{\setof{e \where \f{e}{S'} \leq \epsless \lambda}}}
    &\leq                        %
      \evof{\sizeof{\setof{e \where \f{e}{S} \leq \epsless \lambda}}}
    \\
    &\leq                        %
      \evof{\sizeof{\setof{e \where \f{e}{S'} \leq \epsless \lambda}}}
      +                           %
      \parof{\delta - \delta'} \sizeof{\groundset}^2.
  \end{align*}
\end{lemma}
\begin{proof}
  We apply \reflemma{monotone-estimate-additive} w/r/t the nonnegative,
  monotonic and bounded functions $g(S) = \sizeof{\spn{S}}$ and
  $g(S) = \evof{\sizeof{\setof{e \where \f{e}{S} \leq \epsless
        \lambda}}}$.
\end{proof}

In particular, \reflemma{sensitivity} shows that a change
of $\bigO{\eps / n}$ in $\delta$ changes the expected values by at
most $\bigO{\eps n}$. \reflemma{greedy-probability-search} allows for
additive error to proportional to $\eps n$, so it suffices to search
values of $\delta$ that are about $\eps/n$ apart, as follows.
\begin{lemma}
  \labellemma{greedy-step-discretization} Let $c > 0$ be any fixed
  constant. There is an integer $i \in \naturalnumbers$ such that
  $i \leq \bigO{\frac{k}{\eps}}$ and, for
  $\delta = \frac{i}{4 n} \cdot \eps c$, we have
  \begin{enumerate}
  \item $\evof{\sizeof{\spn{S}}} \leq \frac{3 \eps}{4} \sizeof{\groundset}$
  \item
    \begin{math}
      \evof{\sizeof{\setof{e \where \f{e}{S} \leq \epsless
            \groundset}}} \leq %
      \frac{\parof{3 + c} \eps}{4} \sizeof{\groundset}
    \end{math}
  \item either
    \begin{math}
      \evof{\sizeof{\spn{S}}} \geq \frac{\parof{3 - c}\eps}{4} \sizeof{\groundset}
    \end{math}
    or
    \begin{math}
      \evof{\sizeof{\setof{e \where \f{e}{S} \leq \epsless
            \groundset}}} \geq  \cless \frac{\eps}{4}
      \sizeof{\groundset}
    \end{math}
  \end{enumerate}
\end{lemma}
\begin{proof}
  Let $\optgs$ be the maximum greedy step size satisfying
  \refstep{gms-step-size-margin} and \refstep{gms-step-size-span}
  except with $\eps$ replace by $3 \eps / 4$.  That is, for
  $\optS \sim \optgs \groundset$, we have
  \begin{enumerate}[label={(\roman*)}]
  \item
    \begin{math}
      \evof{\spn{\optS}} \leq \frac{3 \eps}{4} \sizeof{\groundset},
    \end{math}
  \item
    \begin{math}
      \evof{\sizeof{e \where \f{e}{\optS} \leq \epsless \lambda}} \leq
      \frac{3\eps}{4} \sizeof{\groundset},
    \end{math}
    and
  \item
    \begin{math}
      \max{\evof{\sizeof{\spn{\optS}}}, \evof{\sizeof{e \where
            \f{e}{\optS} \leq \epsless \lambda}}} = \frac{3 \eps}{4}
      \sizeof{\groundset}
    \end{math}
  \end{enumerate}
  By \refremark{step-size-range}, we know that
  $\optgs \leq \bigO{k / n}$.

  Let $i = \rounddown{\frac{\optgs}{c \eps / 4n}}$, let
  $\delta = \frac{i c \eps}{4n}$, and let $S \sim \delta \groundset$. We
  have $i \leq \bigO{k / \eps}$ since $\optgs \leq \bigO{k/n}$, and we
  also have
  \begin{align*}
    \optgs - \frac{c \eps}{4n} < \delta \leq \optgs.
  \end{align*}
  By monotonicity, we have
  \begin{math}
    \evof{\spn{S}} \leq \frac{3 \eps}{4} \sizeof{\groundset}
  \end{math}
  and
  \begin{math}
    \evof{\sizeof{e \where \f{e}{S} \leq \epsless \lambda}} \leq
    \frac{3 \eps}{4} \sizeof{\groundset}.
  \end{math}
  By \reflemma{sensitivity}, we have
  \begin{align*}
    &\max{\evof{\sizeof{\spn{S}}}, \evof{\sizeof{e \where \f{e}{S} \leq
      \epsless \lambda}}}
    \\
    &\geq
      \max{\evof{\sizeof{\spn{S}}}, \evof{\sizeof{e \where \f{e}{S} \leq
      \epsless \lambda}}} - \frac{c \eps \sizeof{\groundset}}{4}
      \geq
      \cless \frac{3\eps}{4} \sizeof{\groundset},
  \end{align*}
  as desired.
\end{proof}

\subsubsection{Putting it all together}

We now put everything together and prove
\reflemma{greedy-probability-search}.

\begin{proof}[Proof of \reflemma{greedy-probability-search}]
  Let $c > 0$ be a sufficiently small constant to be determined later.
  For each integer $i \in \naturalnumbers$, let
  $\delta_i = \frac{i c}{4 \eps n}$ and let
  $S_i \sim \delta_i \groundset$. For each positive integer $i$ with
  $i \leq \bigO{\frac{k}{\eps}}$, in parallel, we apply
  \reflemma{span-concentration} and \reflemma{margin-concentration}
  with $\bigO{\log{n}/\eps^2}$ independent samples to obtain
  $\parof{1 - c}$-multiplicative, $c \eps n$-additive approximations
  with probability $1 - 1 / \poly{n}$ to the quantities
  $\evof{\sizeof{\spn{S_i}}}$ and
  \begin{math}
    \evof{\sizeof{e \where \f{e}{S_i} \leq \epsless \lambda}}.
  \end{math}
  Since there are only $\bigO{\frac{k}{\eps}}$ indices, by the union
  bound, they all succeed with probability $1 - 1/\poly{n}$.

  Suppose all the estimates succeed.  By
  \reflemma{greedy-step-discretization}, there exists an index
  $i \leq \bigO{\frac{k}{\eps}}$ such that
  \begin{align*}
    \max{\evof{\sizeof{\spn{S_i}}}, \evof{\setof{e \where \f{e}{S_i}} %
    \leq                        %
    \epsless \lambda}}          %
    \in
    \bracketsof{1-c, 1} \frac{3\eps}{4} \sizeof{\groundset},
  \end{align*}
  which implies that the estimated maximum is in the range
  \begin{align*}
    \bracketsof{\cless^2, 1} \frac{3 \eps}{4} \sizeof{\groundset} \pm
    \eps c \sizeof{\groundset}  %
    =                           %
    \apxcpm \frac{3 \eps}{4} \sizeof{\groundset}.
  \end{align*}
  Any index $j$ where the maximum is estimated to be in the above
  range, for sufficiently small constant $c$, has
  \begin{align*}
    \max{\evof{\sizeof{\spn{S_j}}}, \evof{\setof{e \where \f{e}{S_j}} \leq
    \epsless \lambda}} \in \apxcpm \frac{3 \eps}{4},
  \end{align*}
  which for sufficiently small $c$ gives the desired result.
\end{proof}

\begin{remark}
  The step size $\delta$ chosen by
  \reflemma{greedy-probability-search} is more conservative then the
  step size defined by \algo{greedy-sample}. Because $\delta$ still
  satisfies the inequalities in steps \refstep{gms-step-size-margin}
  and \refstep{gms-step-size-span}, the randomized sets $(I,S)$
  produced by this choice of $\delta$ still forms a greedy
  block. Because $\delta$ is not exactly maximal, we do not expect the
  decrease $\sizeof{\groundset}$ by $\eps \sizeof{\groundset}$ in
  expectation, but it is large enough that we expect to decrease
  $\sizeof{\groundset}$ by $\frac{\eps}{2} \sizeof{\groundset}$. The
  slightly weaker decay rate increases the expected number of
  iterations in \algo{greedy-blocks} by only a constant factor.
\end{remark}

\subsection{Estimating the multilinear relaxation}

The amplification step generates auxiliary submodular functions via
the multilinear extension of the original submodular function, which
applies the original submodular function to a distribution of random
samples. In many cases of practical interest, such as coverage or
graph cuts, the multilinear extension can be computed exactly. In the
oracle model, one can still estimate the multilinear extension
(pointwise) up to sufficient accuracy by repeated sampling. Let
$\ell = \bigO{\reps}$ be the number of iterations in the
amplification, let
\begin{math}
  \opt = \max_{I \in \independents} \f{I}.
\end{math}

We first show that the auxiliary functions $g$ can be estimated up to
$\frac{\poly{\eps}\opt}{k}$ additive error with
$\bigO{n \poly{k,\log n, \eps^{-1}}}$ samples and oracle
queries\footnote{Those only interested in adaptivity might observe
  that since we are averaging over independent samples, the
  multilinear extension can always be computed exactly with $\bigO{1}$
  adaptivity.}. We then show that this precision suffices to apply
\algo{greedy-sample} and \algo{block-greedy}.

\subsubsection{Monotone submodular functions}

In the monotone case, it is fairly easy to estimate the expected
margin of any element by applying Chernoff bounds directly
\cite{bv,cjv}.

\begin{lemma}
  \labellemma{monotone-ml-concentration} Let $\defnnsubmodular$ be a
  monotone submodular function, and let
  $\alpha = \max_{e \in \groundset} f(e)$. Let $\eps > 0$ be
  sufficiently small. For any fixed distribution over sets
  $S \subseteq \groundset$ and fixed $e \in \groundset$, with
  $\bigO{k \log{n} / \eps^3}$ random samples of $S$, one can compute a
  value $X$ such that
  \begin{align*}
    \probof{                                                  %
    \absvof{X - \ef{e}{S}}                                    %
    \geq                                                      %
    \bigOmega{\eps^2 \opt / k}  %
    } %
    \leq                        %
    \frac{1}{\poly{n}}
  \end{align*}
\end{lemma}

\begin{lemma}
  \labellemma{monotone-concentration} For any sufficiently small
  constant $c > 0$, the auxiliary function $g$ in
  \algo{monotone-{\allowbreak}ML-{\allowbreak}amp} can be estimated
  with additive error $c \eps^2 \opt / k$ with
  $\bigO{k \log{n} / \eps^3}$ random samples and oracle queries.
\end{lemma}

\subsubsection{Nonnegative submodular functions}

We now consider estimating generally nonnegative submodular
functions. Here the margin computations are a little trickier because
they are not necessarily nonnegative. We note that obtaining an oracle
complexity that depends on $n$ is easy; we work a little bit harder to
get oracle complexities on the order of $\poly{k}$.

\begin{lemma}
  \labellemma{nn-ml-concentration} Let $\defnnsubmodular$ be a
  nonnegative submodular function, and consider a fixed distribution
  of random sets $S \subseteq \groundset$ where each element is
  sampled independently and $\evof{\sizeof{S}} \leq \bigO{k}$. With
  $\bigO{k^4 \log{n} / \eps^2}$ random samples of $S$, one can compute
  a value $X$ such that
  \begin{align*}
    \probof{\absvof{X - \ef{S}} \geq \bigOmega{\eps^2 \opt / k}}
    \leq
    \frac{1}{\poly{n}}.
  \end{align*}
\end{lemma}


\begin{proof}[Proof]
  Since $k \geq \bigO{\log{n}/\eps^2}$, we have
  $\sizeof{S} \leq \bigO{k}$ with high probability and the probability
  dropping off exponentially beyond $\bigO{k}$.  Since
  \begin{math}
    \f{S} \leq \sizeof{S} \opt,
  \end{math}
  $\f{S}$ is likewise concentrated below $\bigO{\opt k}$, and the
  claim follows from extended formulations of Chernoff inequality that
  treat $\f{S}$ as if it were deterministically bounded above by
  $\bigO{\opt k}$. The argument is standard, but as a demonstrative
  example, consider the moment generating function of $\f{S}$.  Let
  $K = c \evof{\sizeof{S}} = \bigO{k}$ be a constant factor greater
  than $\evof{\sizeof{S}}$, for some suitably large constant $c > 1$
  such that $\probof{\sizeof{S} \geq K} \leq \frac{1}{\poly{n}}$ (for
  some fixed polynomial $\poly{n}$).  For $t > 0$ sufficiently small
  and $L = K \opt$, the moment generating function of $\f{S} / L$ can
  be divided up as
  \begin{align*}
    \evof{\exp{\frac{t \f{S}}{L}}}
    &=                           %
      M_1 + M_2
  \end{align*}
  where
  \begin{align*}
    M_1 &= \evof{\exp{\frac{t \f{S}}{L}} \given \sizeof{S} \leq K} %
          \probof{\sizeof{S} \leq K},
  \end{align*}
  and
  \begin{align*}
    M_2 &=
          \evof{\exp{t \f{S}/L} \given \sizeof{S} > K} %
          \probof{\sizeof{S} > K}.
  \end{align*}
  For the first term, we have
  \begin{align*}
    M_1
    &=\evof{\exp{\frac{t \f{S}}{L}} \given \sizeof{S} \leq K}
      \probof{\sizeof{S} \leq K}
    \\
    &\tago{\leq}                        %
      \evof{1 + \parof{e^t - 1} \f{S} / L \given \sizeof{S} \leq K}
      \probof{\sizeof{S} \leq \bigO{k}} %
    \\
    &\tago{\leq}                       %
      1 + \parof{e^{t} - 1} \evof{\f{S}} / L
  \end{align*}
  by \tagr convexity and that
  \begin{math}
    \f{S} / L \leq \sizeof{S} \opt / L \leq 1
  \end{math}
  by submodularity, and \tagr nonnegativity of $f$. For the second
  term, we have
  \begin{align*}
    \evof{\exp{\frac{t \f{S}}{L}} \given \sizeof{S} \geq K} %
    &\tago{\leq}                                                 %
      \evof{\exp{\frac{t \sizeof{S} \opt}{L}} \given \sizeof{S} \geq K}
    \\
    &\tago{\leq}                        %
      \evof{\exp{\frac{t\sizeof{S}}{K}} \given \sizeof{S} \geq K} %
      \tago{\leq}                                                 %
      e^{t} \evof{\exp{\frac{\sizeof{S}}{K}}} %
    \\
    &\tago{\leq}                        %
      e^t e^{\parof{e^t - 1} \evof{\sizeof{S}} / K}
      \tago{\leq}                        %
      e^{t + e^t  -1} \numberthis.
  \end{align*}
  where \tagr is by submodularity, \tagr cancels like terms, \tagr
  \labeltag{moment-analysis} is derived below, \tagr applies the
  standard moment analysis of the independent sum $\sizeof{S} / K$,
  and \tagr observes that $\evof{\sizeof{S}} \leq K$.  The inequality
  \reftag{moment-analysis} is intuitive and can be proven by standard
  techniques as follows.  Enumerate
  $\groundset = \setof{e_1,\dots,e_n}$. For each $i$, let
  $X_i \in \setof{0,1}$ be an independent indicator variable with
  $\evof{X_i} = \probof{e_i \in S}$. Then
  $\sizeof{S} = \sum_{i=1}^n X_i$ distributionally.  For each $i$, let
  $E_i$ be the event that $X_1 + \cdots + X_i > K$ and
  $X_1 + \cdots + X_{i-1} \leq K$. The events $E_1,\dots,E_n$
  partition the event that $\sizeof{S} > K$. We have
  \begin{align*}
    \evof{\exp{\frac{t \sizeof{S}}{K}} \given \sizeof{S} > K}
    &\tago{=}                           %
      \sum_i \probof{E_i \given \sizeof{S} > K} \evof{\exp{\frac{t \sizeof{S}}{K}} \given E_i}
    \\
    &\tago{=}                          %
      \sum_i \probof{E_i \given \sizeof{S} > K}
      \evof{\exp{\frac{t}{K}\parof{K + \sum_{j > i} X_j}}}
    \\
    &\tago{=}                          %
      e^t \sum_i \probof{E_i \given \sizeof{S} > K}
      \evof{\exp{\frac{t}{K}\parof{\sum_{j > i} X_j}}}
    \\
    &\tago{\leq}                %
      e^t \sum_i \probof{E_i \given \sizeof{S} > K}
      \evof{\exp{\frac{t \sizeof{S}}{K}}}
    \\
    &\tago{=}
      e^t \evof{\exp{\frac{t \sizeof{S}}{K}}}
  \end{align*}
  by \tagr taking conditional probabilities, \tagr independence of the
  $X_i$'s, \tagr nonnegativativity of the $X_i$'s, and \tagr summing
  the probabilities to 1. This proves \reftag{moment-analysis}.
  Plugging \reflastequation into $M_2$, we have
  \begin{align*}
    M_2 \leq                    %
    e^{t + e^t - 1} \probof{\sizeof{S} \geq \bigO{k}} \leq \frac{1}{\poly{n}}.
  \end{align*}
  for sufficiently small $t$.  Thus
  \begin{align*}
    \evof{\exp{\frac{t \f{S}}{L}}}
    &=                          %
      M_1 + M_2                 %
      \leq
      1 + \parof{e^t - 1} \parof{\ef{S} / L + \poly{1/n}}
    \\
    &\leq                      %
      \exp{\parof{e^t - 1} \frac{\ef{S}}{L}} + \poly{1/n}.
  \end{align*}
  This is essentially the moment inequality one would obtain if
  $\ef{S}$ was bounded above by $\bigO{k \opt}$ deterministically, as
  the $\poly{1/n}$ term is negligible. Plugging in to the proof of the
  Chernoff inequality w/r/t the upper tail yields the desired result.
  Similarly, one can show that for sufficiently small $t > 0$,
  \begin{align*}
    \evof{\exp{\frac{-t \f{S}}{L}}}
    \leq                        %
    \exp{\parof{e^{-t} - 1}\frac{\ef{S}}{L}} + \poly{1/n},
  \end{align*}
  which leads to the concentration bound on the lower tail.
\end{proof}

\begin{lemma}
  \labellemma{nn-amp-concentration} The auxiliary function $g$ in
  \algo{nn-ML-amp} can be estimated with additive precision
  $c \eps^2 \opt / k$ using $\bigO{k^4 \log{n} / \poly{\eps}}$ queries
  for any fixed constant $c$.
\end{lemma}
\begin{proof}
  The choice of $\delta$ always ensures the independent sample $S$ in
  \algo{greedy-sample} expects to take $\bigO{k}$ elements, since part
  (b) in the proof of \reftheorem{greedy-sample} shows that we expect
  to prune only an $\eps$-fraction of initially sampled elements for
  the sake of independence. The auxiliary function $g$ takes a uniform
  $1/\ell$-probability independent sample of $S'$ with a uniform
  $1/\ell$-probability samples of $m \leq \ell$ independent sets
  $I_1,\dots,I_{m}$, so the expected size of the random set submitted
  to $f$ is $\bigO{k}$.
\end{proof}

\begin{remark}
  We believe the oracle complexity can be reduced to
  \begin{math}
    \bigO{k^2 \poly{\log n, \eps^{-1}}},
  \end{math}
  which we plan to address in future work.
\end{remark}

\subsubsection{On the robustness of \algo{greedy-sample}}

For either auxiliary function $g$ in the amplification framework, we
have $g(e) \leq \bigO{\eps f(e)}$ for all $e \in \groundset$. We also
have $g(T) \geq \bigOmega{\eps f(T)}$ for all $T \in
\independents$. One can show that for a sufficiently good
approximation ratio within the amplification framework, it suffices to
check values of $\lambda$ that are $\geq c \eps^2 \opt / k$ for some
constant $c > 0$. (Note that this value may be higher than
$\max_e \eps g(e) / n$, in which case we are truncating the loop in
\algo{block-greedy} to terminate when
$\lambda \leq c \eps^2 \opt / k$).

It is important to note that \algo{greedy-sample} and the general
thresholding framework allows estimation errors relative to the
threshold values $\lambda$, as follows. The following guarantees are a
little weaker then \reftheorem{greedy-sample}, but can be shown to be
sufficient w/r/t the \algo{block-greedy} framework by retracing the
proofs. .
\begin{lemma}
  \labellemma{amp-greedy-sample} Consider the setting of
  \reftheorem{greedy-sample} for
  $\lambda \geq \bigOmega{\eps^2 \opt / k}$, applied to the implicit
  function $g$ induced by either \algo{monotone-ML-amp} or
  \algo{nn-ML-amp}. With $\bigO{1}$ adaptivity and
  $\bigO{n \poly{k,\log n,\eps^{-1}}}$ oracle calls, and with high
  probability, one can compute sets
  $I, S, \groundset' \subseteq \groundset$ such that:
  \begin{mathresults}
  \item $(I,S)$ is a $\parof{1 - c \eps}$-greedy block.
  \item                         %
    \begin{math}
      \setof{e \notin \spn{S} \where \g{e}{S} \geq \parof{1 - \eps}
        \lambda} \subseteq \groundset' \subseteq \setof{e \notin
        \spn{S} \where \g{e}{S} \geq \parof{1 - c \eps}\lambda}.
    \end{math}
  \item
    $\evof{\sizeof{\groundset'}} \leq \parof{1 - d \eps}
    \sizeof{\groundset}$.
  \end{mathresults}
  Here $c > 1$ and $d > 0$ are any desired and sufficiently small
  constants.
\end{lemma}

\begin{proof}[Proof sketch]
  At a high level, \reflemma{monotone-concentration} and
  \reflemma{nn-amp-concentration} gives us approximations to the value
  of any particular set of the marginal value of any element w/r/t any
  particular set with additive error at most $c \eps \lambda$, for any
  desired constant $c > 0$, with high probability. This introduces
  some error proportional to $c \eps \lambda$. The additive error
  carries through the proof and does not change things substantially.
  We restrict ourselves to a sketch.

  The following sublemma states that we can still approximate the
  greedy parameter $\delta > 0$.
  \begin{sublemma}
    \labellemma{apx-greedy-probability-search} Let $c > 0$ be any
    sufficiently small constant. With high probability, with
    $\bigO{\poly{k, \log{n}, \eps^{-1}}}$ oracle queries and
    $\bigO{1}$ adaptivity, we can find a sampling probability $\delta$
    such that for $S \sim \delta \sizeof{\groundset}$,
    \begin{mathresults}
    \item
      \begin{math}
        \evof{\spn{S_i}} \leq \frac{3 \eps}{4} \sizeof{\groundset}.
      \end{math}
    \item
      \begin{math}
        \evof{\sizeof{\setof{e \where \g{S}{e} \leq \epsless \lambda - c
              \eps \lambda}}} \leq \frac{3 \eps}{4} \sizeof{\groundset}.
      \end{math}
    \item
      \begin{math}
        \max{\spn{S_i}, \evof{\sizeof{\setof{e \where \g{S}{e} \leq
                \epsless \lambda + c \eps \lambda}}} } \geq %
        \cless \frac{3\eps}{4} \sizeof{\groundset}.
      \end{math}
    \end{mathresults}
  \end{sublemma}
  \begin{subproof}[Proof sketch]
    By either \reflemma{monotone-concentration} or
    \reflemma{nn-amp-concentration}, we can estimate each $g_e(S')$
    for any fixed set $e'$ up to an $c \eps \lambda$ additive error
    with $\bigO{k\log{n}/\eps^2}$ samples and queries. Retracing the
    proof of \reflemma{greedy-probability-search}, these
    approximations introduce the additive $c \eps \lambda$ in (ii) and
    (iii) (up to constant factors).
  \end{subproof}
  Let $c > 0$ be a small constant to be determined later, and let
  $\delta$ be the sampling probability produced by
  \reflemma{apx-greedy-probability-search}. We first sample
  $S \sim \delta \groundset$. For each $e \in S$, by either
  \reflemma{monotone-concentration} or
  \reflemma{nn-amp-concentration}, we compute
  $\frac{c \eps}{k}$-additive approximations to $\g{e}{S-e}$, denoted
  $\apxg{s}{S-e}$, which all succeed with probability
  $1 - 1/\poly{n}$. We set
  \begin{math}
    I = \setof{e \in S \where \apxg{e}{S-e} \geq \epsless \lambda - 2 c
      \eps \lambda}.
  \end{math}
  We set
  $\groundset' = \setof{e \in \groundset \setminus \spn{S} \where
    \apxg{e}{S-e} \leq \epsless \lambda + 2 c \eps \lambda}$. One now
  retraces the proof of \reftheorem{greedy-sample} to obtain the
  desired result (for sufficiently small $c > 0$).
\end{proof}


              %

\bibliographystyle{plainnat} %
\bibliography{parallel-submodular-matroids} %

\appendix

\section{Preliminaries}

\labelappendix{preliminaries}

This paper is primarily concerned with two abstract objects,
submodular functions and matroids, which we define formally in
\refappendix{submodular} and \refappendix{matroids} respectively.

\subsection{Submodular functions}
\labelappendix{submodular}

Let $\defrsubmodular$ be a real-valued set function. $f$ is:
\begin{enumerate}
\item \defterm{normalized} if $f(\emptyset) = 0$.
\item \defterm{nonnegative} if $f(S) \geq 0$ for all $S \subseteq \groundset$.
\item \defterm{monotone} $f(S) \leq f(T)$ for all $S \subseteq T
  \subseteq \groundset$.
\item \defterm{submodular} if
  $f(S) + f(T) \geq f(S \cup T) + f(S \cap T)$ for all $S \subseteq T
  \subseteq \groundset$.
\end{enumerate}
All set valued functions in this paper are normalized and submodular,
and submodular functions will always be assumed to be normalized. The
given function that we are optimizing over is nonnegative, but other
set-valued functions arise that are not necessarily nonnegative.

Submodularity can be understood intuitively in terms of ``decreasing
marginal returns''. To this end, we denote
\begin{center}
  \begin{math}
    f_S(U) \defeq f(S \cup U) - f(S)
  \end{math}
  for $S, U \subseteq \groundset$.
\end{center}
$f_S(U)$ is represents the increase in value gained by adding $U$ to
$S$, and more succinctly called the \defterm{marginal value of $U$ to
  $S$}. Submodularity (as defined above) is equivalent to saying that
the marginal value of $U$ is decreasing in $S$ in the following sense:
\begin{center}
  \begin{math}
    f_S(U) \geq f_T(U)
  \end{math}
  for all
  \begin{math}
    S,T,U \subseteq \groundset
  \end{math}
  with
  \begin{math}
    S \subseteq T.
  \end{math}
\end{center}
For any submodular function $f$ and set $S$, the marginal values
$f_S: \groundset \to \reals$ form a normalized submodular function. If
$f$ is monotone, then $f_S$ is monotone and nonnegative. However, $f$
being nonnegative does \emph{not} imply that $f_S$ is nonnegative.

We appeal to a continuous extension of $f$ to $\nnreals^n$ called the
``multilinear extension'', for which we introduce the following
notation. For a vector $x \in \nnreals^n$, we write $S \sim x$ to
denote the random set $S \subseteq \groundset$ that samples each
$e \in \groundset$ independently with probability $\min{x_e, 1}$. That
is, we interpret $x$ (after truncation) as the margins of an
independent sample. The \defterm{multilinear extension of $f$},
denoted $F: \nnreals^{\groundset} \to \reals$, is the expected value
of a random set drawn according to $x$:
\begin{align*}
  \mlf{x} = \evof{f(S) \given S \sim x}.
\end{align*}

The name ``multilinear extension'' can be explained as follows. We
identify each set $S$ with its incidence vector in
$\setof{0,1}^{\groundset}$.  Abusing notation, we let $S$ also denote
its incident vector (when the meaning is clear). Then we have
$\mlf{S} = f(S)$ for every set $S$, and $\mlf$ is an extension of $f$
as a function of $\setof{0,1}^{\groundset}$. The ``multilinear'' comes
from the fact that $\mlf{x}$ is multilinear in $x$ when
$x \in \groundcube$.

As an expectation of $f$, $\mlf$ inherits many of the properties as
$f$. $f$ is normalized, nonnegative, and monotone iff $\mlf$ is
normalized, nonnegative and monotone, respectively. Submodularity
translates to a particular kind of concavity, as follows. We say that
$\mlf$ is \defterm{monotone concave} if for any
$x, v \in \nnreals^{\groundset}$ and $\delta > 0$, the map
\begin{center}
  $\delta \mapsto \mlf{x + \delta v}$ is concave in $\delta$.
\end{center}
Then $f$ is submodular iff $\mlf$ is monotone concave. For example, if
$f$ is submodular and monotonic, then
\begin{center}
  \begin{math}
    \mlf{\eps x} \geq \eps \mlf{x}
  \end{math}
  for all $x \in \nnreals^{\groundset}$ and $\eps \in (0,1)$.
\end{center}

The following useful property observed by \citet{bfns} can be proven via a
different extension of $f$ called the \emph{Lovász extension}.
\begin{lemma}[{\citealp{bfns}}]
  \labellemma{lovasz-floor} Let $f$ be a nonnegative submodular
  function. Let $S,T \subseteq \groundset$ where $S$ is a random
  set. Let $\delta = \max_e \probof{e \in S}$ by the maximum
  probability of any set. Then
  \begin{align*}
    \evof{f(S \cup T)} \geq (1-\delta) f(T).
  \end{align*}
\end{lemma}

\subsection{Combinatorial constraints}
\labelappendix{matroids}

A set system $\defmatroid$ consists of a \defterm{ground set}
$\groundset$ and a family of subsets
$\independents \subseteq \groundsets$. A set system $\defmatroid$ is
an \defterm{independence system} if $\independents$ is nonempty and
closed under subsets:
\begin{enumerate}[{label={(\alph*)}}]
\item $\emptyset \in \independents$
\item For $S \subseteq T \subseteq \groundset$, $T \in \independents$
  implies $S \in \independents$.
\end{enumerate}
A set $S \in \independents$ is called an \defterm{independent set}.
An independence system is a \defterm{matroid} if it also satisfies the
following augmentation property.
\begin{enumerate}[resume, label={(\alph*)}]
\item If $S,T \in \independents$ and $\sizeof{S} < \sizeof{T}$, then
  there is an element $e \in T \setminus S$ such that
  $S + e \in \independents$.
\end{enumerate}
A base is a maximal independent set. By property (c), every base in a
matroid has the same cardinality, called the \defterm{rank}. More
generally, for any set $S \subseteq \groundset$ in a matroid
$\defmatroid$, every maximal independent set in $S$ has the same
cardinality, called the \defterm{rank of $S$} and denoted
$\rank{S}$. The rank is a nonnegative, monotone submodular function.
The \defterm{span} of an independent set $I \in \independents$ is the
set of elements $e \in \groundset$ such that either $e \in I$ or
$I + e \notin \independents$, where $I+ e$ is a shorthand for
$I \cup \setof{e}$. In general, the span of a set $S$ is the set of
elements that do not increase the rank:
\begin{align*}
  \spn{S} = \setof{e \in \groundset \where \rank{S + e} = \rank{S}}.
\end{align*}

Given a matroid $\defmatroid$, there are two different ways to modify
$\matroid$ of interest. Given a set $S \subseteq \groundset$, the
\defterm{restriction of $\matroid$ to $S$}, denoted
$\matroid \land S = \parof{S, \independents \land S}$, has ground set
$S$ and independent sets consisting of the independent subsets of $S$,
\begin{align*}
  \independents \land S = \setof{I \subseteq S \where I \in \independents}.
\end{align*}
The rank of $\matroid \land S$ is precisely the rank of $S$.

The second modification is contraction. Give a set $S \subseteq
\groundset$, the \defterm{contraction of $\matroid$ to $S$}, denoted
$\matroid / S = (\groundset / S, \independents / S)$, has ground set
\begin{align*}
  \groundset / S = \groundset \setminus \spn{S},
\end{align*}
and independence is defined by the rank function
\begin{align*}
  \rank_S(T) = \rank{S + T} - \rank{S}.
\end{align*}
Alternatively, one can choose any base $B$ of $S$, and define
$\independents/ S$ by
\begin{align*}
  \independents_S = \setof{T \subseteq \groundset / S \where B \cup T
  \in \independents}.
\end{align*}

A good working example of a matroid is the \defterm{graphic
  matroid}. Here the ground set corresponds to the edges of some graph
$G = (V,E)$, and a set of edges is independent if they form a
forest. The bases of this matroid are the spanning forests in the
graph. Restricting a graphic matroid to a set of edges is the graphic
matroid over the subgraph induced by these edges. Contracting a set of
edges corresponds to the graphic matroid over the minor obtained by
contracting each of these edges.

A particular useful property of matroids, first observed by
\citet{brualdi}, is the following.
\begin{lemma}
  Let $I$, $J$ be two independent sets with $\sizeof{I} \geq
  \sizeof{J}$. Then there exists an injection $\pi: J \setminus I \to I
  \setminus J$ such that for all $e \in J$,
  \begin{align*}
    I - \pi(e) + e \in \independents.
  \end{align*}
\end{lemma}

Brualdi's exchange mapping easily implies the following, which is in a
slightly more convenient form for us.
\begin{lemma}
  \labellemma{brualdi-mapping-span}
  Let $S_1,S_2,\cdots,S_k$ be a sequence of sets in a matroid
  $\defmatroid$ such that
  \begin{center}
    \begin{math}
      S_i \subseteq \groundset \setminus \spn{\uS{i-1}}
    \end{math}
    for each $i$.
  \end{center}
  For any independent set $I \in \independents$, one can partition $T
  \cap \spn{\uS{k}}$ into sets $\setof{T_1,\dots,T_k}$ such that for
  each $i$,
  \begin{center}
    \begin{math}
      T_i \subseteq \groundset \setminus \spn{\uS{i}}
    \end{math}
    and
    \begin{math}
      \sizeof{T_i} \leq \sizeof{S_i}.
    \end{math}
  \end{center}
\end{lemma}

\subsubsection{Combinations of matroids}

\labelappendix{matroid-combinations}

We optimize over matroids and combinations of matroids, such as
intersections of matroids and more generally matchoids. A
\defterm{matroid intersection} is an independence system $\defmatroid$
where $\independents = \bigcap_{i} \independents_i$ for a collections
matroids $\matroid_i = (\groundset, \independents_i)$ with the same
set. Matroid intersections generalize bipartite matchings and
arboresences. A \defterm{matchoid} $\defmatroid$ is an independence
system defined by a collection of matroids
$\matroid_i = (\groundset_i, \independents_i)$, where
$\groundset_i \subseteq \groundset$ for each $i$, with independent
sets
\begin{align*}
  \independents = \setof{S \where S \cap \groundset_i \in
  \independents_i \text{ for all } i}.
\end{align*}
For $k \in \naturalnumbers$, a $k$-matroid intersection is an
intersection of $k$ matroids, and a $k$-matchoid is a matchoid where
each element $e \in \groundset$ participates in at most $k$ of the
underlying matroids. A $k$-matroid is of course a $k$-matchoid. A
$k$-matchoid can be recast as a matroid intersection by extending each
underlying $\matroid_i = (\groundset_i,\independents_i)$ to all of
$\groundset$ by allowing any extra element
$e \in \groundset \setminus \groundset_i$ to only be spanned by
itself. The number of matroids in the matroid intersection may be much
larger than $k$, which matters only because the approximation ratios
depend on $k$.

Matroid intersections and matchoids still carry some of the structure
and notions of matroids. Suppose a matchoid or matroid intersection is
defined by the matroids $\matroid_1,\dots,\matroid_k$. One can define
a function $\spn: \groundsets \to \groundsets$ by
\begin{align*}
  \spn{S} = \bigcup_{i=1}^k \spn_i(S)
\end{align*}
where for each $i$, $\spn_i$ is the span function associated with the
$i$th matroid. This span function still has the following properties
of span functions for matroids:
\begin{enumerate}
\item $\spn{S} \subseteq \spn{T}$ for $S \subseteq T$.
\item If $S \subseteq \groundset$ and $e \notin \spn{S - e}$ for all
  $e \in S$, then $S \in \independents$.
\end{enumerate}
There notions of restricting and contract a matroid intersection or
matchoid are still well-defined, by taking the restriction or
contraction in each of the underlying matroids, and recombining the
restricted or contracted matroids into a matroid intersection or
matchoid.

Canonical examples of 2-matroid intersection are bipartite matchings
and arborescences. An example of a 2-matchoid is a matching, and an
example of a $k$-matchoid is a matching in a hypergraph of rank
$k$. In a matching (bipartite or general), the span of an edge set $S$
is the set of all edges incident to some edge in $S$. Contracting an
edge corresponds to removing both endpoints and all incident edges.

Brualdi's exchange theorem extends to $p$-matchoids as follows.
\begin{lemma}
  \labellemma{k-matroids-exchange} Let $\defmatroid$ be a
  $p$-matchoid.  Let $S_1,S_2,\dots,S_k$ be a sequence of sets such
  that for each $i$, $S_i \subseteq \groundset \setminus
  \spn{S_i}$. For any independent set $T \in \independents$, one can
  partition $T \cap \spn{\uS{k}}$ into sets $\setof{T_1,\dots,T_k}$
  such that for each $i$,
  \begin{center}
    \begin{math}
      T_i \subseteq \groundset \setminus \spn{\uS{i}}
    \end{math}
    and
    \begin{math}
      \sizeof{T_i} \leq p \sizeof{S_i}.
    \end{math}
  \end{center}
\end{lemma}


\end{document}